\documentclass[12pt,a4paper]{article}%
\usepackage{a4}
\usepackage{amsmath}
\usepackage{graphicx}
\usepackage{amsfonts}
\usepackage{amssymb}%
\newtheorem{theorem}{Theorem}
\newtheorem{lemma}[theorem]{Lemma}
\newtheorem{definition}[theorem]{Definition}

\newenvironment{proof}[1][Proof]{\textbf{#1.} }{\ \rule{0.5em}{0.5em}}
\begin{document}

\title{Exact form factors of the $SU(N)$ Gross-Neveu model and $1/N$ expansion}
\author{Hrachya M. Babujian%
\thanks{Address: Alikhanian Brothers 2, Yerevan, 375036Armenia}
\thanks{E-mail: babujian@yerphi.am} , Angela Foerster%
\thanks{Address: Instituto de F\'{\i}sica da UFRGS, Av. Bento Gon\c{c}alves 9500,
Porto Alegre, RS - Brazil}
\thanks{E-mail: angela@if.ufrgs.br} , and Michael Karowski%
\thanks{E-mail: karowski@physik.fu-berlin.de}\\
Institut f\"{u}r Theoretische Physik, Freie Universit\"{a}t Berlin,\\
Arnimallee 14, 14195 Berlin, Germany }
\date{\today
\vskip6mm \textit{{\small Dedicated to the memory of Alexey Zamolodchikov}}}
%\date{January 3, 2007 }
\maketitle

\begin{abstract}
The general $SU(N)$ form factor formula is constructed. Exact form factors for
the field, the energy momentum and the current operators are derived and
compared with the $1/N$-expansion of the chiral Gross-Neveu model and full
agreement is found. As an application of the form factor approach the equal
time commutation rules of arbitrary local fields are derived and in general
anyonic behavior is found.\\[8pt]
PACS: 11.10.-z; 11.10.Kk; 11.55.Ds \newline
Keywords: Integrable quantum field theory, Form factors

\end{abstract}

\section{Introduction}

Quantum chromodynamics, the theory of the strong interactions, is a non
abelian gauge theory based on the gauge group $SU(3)$. It was first pointed
out by 't Hooft \cite{Hooft1,Hooft2} that many features of QCD can be
understood by studying a gauge theory based on the gauge group $SU(N)$ in the
limit $N$ $\rightarrow\infty$. One might think that letting $N\rightarrow
\infty$ would make the analysis more complicated because of the larger gauge
group and consequently increase in the number of dynamical degrees of freedom.
Also one can think that $SU(N)$ gauge theory has very little to do with QCD
because $N\rightarrow\infty$ is not close to $N=3$. However it is well known
that the $1/N$ expansion provides good results which can be compared with
experiments \cite{Mano}.

One of the most important trends in theoretical physics in the last decades is
the development of exact methods which are completely different from
perturbation theory. Resolution of the strong coupling problem would give us a
full understanding of the structure of interactions in nonabelian gauge
theory. One promising possibility of overcoming the limitations of
perturbation theory is the application of exact integrability. From this point
of view the two dimensional integrable quantum field theories are in a sense a
laboratory for investigations of those properties of quantum field theories,
which cannot be described via perturbation theory.

The chiral $SU(N)$ Gross-Neveu \cite{GN} model given by the Lagrangian
\begin{equation}
\mathcal{L}=\sum_{i=1}^{N}\bar{\psi}_{i}\,i\gamma\partial\,\psi_{i}%
+\frac{g^{2}}{2}\left(  \left(  \sum_{i=1}^{N}\bar{\psi}_{i}\psi_{i}\right)
^{2}-\left(  \sum_{i=1}^{N}\bar{\psi}_{i}\gamma^{5}\psi_{i}\right)
^{2}\right)  \label{1.2}%
\end{equation}
is an interesting $1+1$ dimensional field theory that can be studied using the
$1/N$ expansion. The model is asymptotically free with a spontaneously broken
chiral symmetry, and so shares some dynamical features with QCD. Gross and
Neveu \cite{GN} investigated the\ model using an $1/N$ expansion. Apparently a
chiral $U(1)$-symmetry is spontaneously broken, the fermions acquire mass and
a Goldstone boson seems to appear. This is of course not possible in two
space-time dimensions and severe infrared divergences appear due the
\textquotedblleft would-be-Goldstone boson\textquotedblright. However, it has
been argued by Witten \cite{Wi} that dynamical mass generation can be
reconciled with the absence of spontaneous symmetry breaking. There exist
further (different) approaches to overcome these problems and to formulate a
$1/N$ expansion \cite{KuS,ABW,KKS} (see also \cite{AAR}). On shell they all
agree and are consistent with the exact S-matrix (\ref{2.2}). We follow here
the approach of Swieca et al.~\cite{KKS} where additional fields are
introduced in order to compensate the infrared divergences. The authors claim
that since the physical fermions have lost not only the chiral $U(1)$ but also
the charge $U(1)$ symmetry, they transform accordingly to pure $SU(N)$. They
propose an interpretation of the antiparticles as a bound state of $N-1$
particles. Furthermore this means that the particles satisfy neither Fermi nor
Bose statistics, but rather carry \textquotedblleft spin\textquotedblright%
\ $s=\frac{1}{2}\left(  1-1/N\right)  $. As a consequence there are unusual
crossing relations and Klein factors.

In this article we will focus on the $SU(N)$ Gross-Neveu \underline{form
factors} using the \textquotedblleft bootstrap program\textquotedblright%
\ \cite{K,K2}. We provide here some examples, calculate the form factors
exactly and compare the results with field theoretical $1/N$ expansions. We
emphasize that in addition to the operators in the vacuum sector, such as the
energy momentum tensor and the current, we also consider anyonic operators as
the fundamental fields. We also derive the equal time commutation rules for
local operators which are in particular complicated due to the unusual
crossing formulae related to the Klein factors.

The general form factor of an operator $\mathcal{O}(x)$ for n-particles, which
is a co-vector valued function and can be written as \cite{KW}
\[
F_{\underline{\alpha}}^{\mathcal{O}}(\underline{\theta})=K_{\underline{\alpha
}}^{\mathcal{O}}(\underline{\theta})\prod_{1\leq i<j\leq n}F(\theta_{ij})
\]
where $\underline{\theta}=(\theta_{1},\dots,\theta_{n})$ is the set of
rapidities of the particles $\underline{\alpha}=(\alpha_{1},\dots,\alpha_{n}%
)$. The scalar function $F(\theta)$ is the minimal form factor function and
the K-function $K_{\underline{\alpha}}^{\mathcal{O}}(\underline{\theta})$
contains the entire pole structure and its symmetry is determined by the form
factor equations (i) to (v) \cite{BFK1}. To construct the K-function we must
apply the nested off-shell Bethe ansatz to capture the vectorial content of
the form factors. This solves the missing link of Smirnov's \cite{Sm} formula
for the $SU(N)$ form factors, where the vectors were given by an
\textquotedblleft indirect definition\textquotedblright\ characterized by
necessary properties but not provided explicitly. We note that $SU(N)$ form
factors were also calculated in \cite{Sm,NT,Ta} using other techniques, see
also the related papers \cite{MTV,Pa}. Our results apply not only to
chargeless operators such as the energy momentum and the current operators but
also to more general ones with anyonic behavior. We believe that our integral
representation, besides of being appropriate for a comparison with field
theoretical $1/N$ expansions, may also shed some light for a better
understanding on the correlation functions of models with more general
(anyonic) statistics.

The paper is organized as follows: In section \ref{s2} we present the general
setting concerning the $SU(N)$ S-matrix, the nested off-shell Bethe ansatz and
the chiral Gross-Neveu Lagrangian field theory. We review known results and
derive some further formulae which we need in the following. In section
\ref{s3} we construct the general form factor formula and present some
examples in detail, such as the form factors of the energy-momentum tensor,
the Dirac field and the $SU(N)$ current. In section \ref{s4} we compare our
exact results against $1/N$ perturbation theory of the $SU(N)$ Gross-Neveu
model. In section \ref{s5} we present the commutation rules of the fields. Our
conclusions are stated in section \ref{s6}. In appendix \ref{sb} we provide
the general proof of the bound state form factor formula and in appendix
\ref{sc} the commutation rule of two fields (in general anyonic) is proved.

\section{General setting}

\label{s2}The particle spectrum of the chiral $SU(N)$ Gross-Neveu model
consists of $N-1$ multiplets of particles of mass $m_{r}=m_{1}\sin\left(
r\pi/N\right)  /\sin\left(  \pi/N\right)  $, which correspond to all
fundamental $SU(N)$ representations of rank $r=1,\dots,N-1$ with
representation spaces $V^{(r)}$ of dimension $\binom{N}{r}$. Let
$(\alpha)=(\alpha_{1},\dots,\alpha_{r}),~(1\leq\alpha_{1}<\dots<\alpha_{r}\leq
N)$ be a particle of rank $r$. We write%
\[
(\alpha)\in V=\bigoplus_{r=1}^{N-1}V^{(r)},~V^{(r)}\simeq\mathbb{C}^{\binom
{N}{r}}%
\]
where the $(\alpha)$ form a basis of $V$. A particle of rank $r$ is a bound
state of $r$ particles of rank 1. The antiparticle corresponding to $(\alpha)$
is $(\bar{\alpha})=(\bar{\alpha}_{1},\dots,\bar{\alpha}_{N-r}),~(1\leq
\bar{\alpha}_{1}<\dots<\bar{\alpha}_{N-r}\leq N)$ (of rank $N-r$) such that
the union of the set of indices satisfies $\{\alpha_{1},\dots,\alpha_{r}%
\}\cup\{\bar{\alpha}_{1},\dots,\bar{\alpha}_{N-r}\}=\{1,\dots,N\}$.

\subsection{The S-matrix}

The S-matrix for the scattering of two particles $\alpha,\beta$ (of rank 1)
\cite{BKKW,BW,KuS,KKS} is%
\begin{equation}
S_{\alpha\beta}^{\delta\gamma}(\theta)=\delta_{\alpha}^{\gamma}\delta_{\beta
}^{\delta}b(\theta)+\delta_{\alpha}^{\delta}\delta_{\beta}^{\gamma}c(\theta)
\label{2.2}%
\end{equation}
where $\theta=\theta_{1}-\theta_{2}$ is the rapidity difference and
$p_{1,2}^{\mu}=m\left(  \cosh\theta_{1,2},\sinh\theta_{1,2}\right)  $. The
amplitudes satisfy%
\[
c(\theta)=-\frac{i\eta}{\theta}b(\theta)\,,~~\eta=\frac{2\pi}{N}\,,
\]%
\[
a(\theta)=b(\theta)+c(\theta)=-\frac{\Gamma\left(  1-\frac{\theta}{2\pi
i}\right)  \Gamma\left(  1-\frac{1}{N}+\frac{\theta}{2\pi i}\right)  }%
{\Gamma\left(  1+\frac{\theta}{2\pi i}\right)  \Gamma\left(  1-\frac{1}%
{N}-\frac{\theta}{2\pi i}\right)  }\,.
\]
We also need the S-matrix for the scattering of a bound state $(\rho
)=(\rho_{1},\dots,\rho_{N-1}),~(1\leq\rho_{1}<\dots<\rho_{N-1}\leq N)$ (of
rank $N-1$) and a particle $\alpha$ (of rank 1)%
\begin{equation}
S_{(\rho)\alpha}^{\beta(\sigma)}(\theta)=(-1)^{N-1}\left(  \delta_{(\rho
)}^{(\sigma)}\delta_{\alpha}^{\gamma}\,b(\pi i-\theta)+\mathbf{C}%
^{\beta(\sigma)}\mathbf{C}_{(\rho)\alpha}\,c(\pi i-\theta)\right)  \label{2.4}%
\end{equation}
where the charge conjugation matrices are defined by%
\[%
\begin{array}
[c]{c}%
\mathbf{C}_{(\alpha_{1}\dots\alpha_{N-1})\alpha_{N}}=\mathbf{C}_{\alpha
_{1}(\alpha_{2}\dots\alpha_{N})}=\epsilon_{\alpha_{1}\alpha_{2}\dots\alpha
_{N}}\\[5pt]%
\mathbf{C}^{\alpha_{1(}\alpha_{2}\dots\alpha_{N})}=\mathbf{C}^{(\alpha
_{1}\dots\alpha_{N-1})\alpha_{N}}=(-1)^{N-1}\epsilon^{\alpha_{1}\alpha
_{2}\dots\alpha_{N}}%
\end{array}
\]
with $\epsilon_{\alpha_{1}\dots\alpha_{N}}$ and $\epsilon^{\alpha_{1}%
\dots\alpha_{N}}$ totally anti-symmetric and $\epsilon_{1\dots N}%
=\epsilon^{1\dots N}=1$. Formula (\ref{2.4}) is obtained by applying
iteratively the bound state fusion method \cite{K1} to (\ref{2.2}).

For later convenience we extract the factors $a(\theta)$ and $(-1)^{N-1}%
b(i\pi-\theta)$, respectively%
\begin{align}
S_{\alpha\beta}^{\delta\gamma}(\theta)  &  =a(\theta)\tilde{S}_{\alpha\beta
}^{\delta\gamma}(\theta)\label{2.6}\\
S_{(\rho)\alpha}^{\beta(\sigma)}(\theta)  &  =(-1)^{N-1}b(i\pi-\theta
)\tilde{S}_{(\rho)\alpha}^{\beta(\sigma)}(\theta) \label{2.8}%
\end{align}
such that%
\begin{gather}
\tilde{S}_{\alpha\beta}^{\delta\gamma}(\theta)=\delta_{\alpha}^{\gamma}%
\delta_{\beta}^{\delta}\tilde{b}(\theta)+\delta_{\alpha}^{\delta}\delta
_{\beta}^{\gamma}\tilde{c}(\theta)\label{2.10}\\
\tilde{S}_{(\rho)\alpha}^{\beta(\sigma)}(\omega)=\delta_{(\rho)}^{(\sigma
)}\delta_{\alpha}^{\beta}+\mathbf{C}^{\beta(\sigma)}\mathbf{C}_{(\rho)\alpha
}\tilde{d}(\omega)\label{2.12}\\
\tilde{b}(\theta)=\frac{\theta}{\theta-i\eta}\,,~~\tilde{c}(\theta
)=\frac{-i\eta}{\theta-i\eta},~~\tilde{d}(\omega)=\frac{c(i\pi-\omega)}%
{b(i\pi-\omega)}=\frac{-i\eta}{i\pi-\omega},~\eta=\frac{2\pi}{N}\,.\nonumber
\end{gather}
where $\delta_{(\rho)}^{(\sigma)}=\delta_{\rho_{1}}^{\sigma_{1}}\delta
_{\rho_{2}}^{\sigma_{2}}...\delta_{\rho_{N-1}}^{\sigma_{N-1}}$. Below we will
also use for the matrices (\ref{2.10}) and (\ref{2.12}) the notations
$\tilde{S}_{12}(\theta)$ and $\tilde{S}_{\bar{1}2}(\theta)$, respectively.

\subsection{Nested \textquotedblleft off-shell\textquotedblright\ Bethe
ansatz}

The \textquotedblleft off-shell\textquotedblright\ Bethe ansatz is used to
construct vector valued functions which have symmetry properties according to
a representation of the permutation group generated by a factorizing S-matrix.
In addition they satisfy matrix differential \cite{B1} or difference
\cite{BKZ2} equations. For the application to form factors we use the
co-vector version $K_{_{1\dots n}}(\underline{\theta})\in V_{1\dots n}=\left(
\bigotimes_{i=1}^{n}V\right)  ^{\dag},~(\theta_{i}\in\mathbb{C},\mathbb{~}%
i=1,\dots,n)$%
\begin{align*}
K_{\dots ij\dots}(\dots,\theta_{i},\theta_{j},\dots)  &  =K_{\dots ji\dots
}(\dots,\theta_{j},\theta_{i},\dots)\,\tilde{S}_{ij}(\theta_{ij})\\
K_{1\dots n}(\underline{\theta}^{\prime})  &  =K_{1\dots n}(\underline{\theta
})Q_{1\dots n}(\underline{\theta},i)
\end{align*}
where $\underline{\theta}^{\prime}=(\theta_{1},\dots,\theta_{i}+2\pi
i,\dots,\theta_{n})$ (see below and e.g.~\cite{BKZ2,BFK1}). We write the
components of the co-vector $K_{_{1\dots n}}$ as $K_{\underline{\alpha}}$
where $\underline{\alpha}=((\alpha_{11},\dots,\alpha_{1r_{1}}),\dots
,(\alpha_{n1},\dots,\alpha_{nr_{n}}))$ is a state of $n$ particles of rank
$r_{1},\dots,r_{n}$.

The nested $SU(N)$ \textquotedblleft off-shell\textquotedblright\ Bethe ansatz
for particles of rank $1$ has been constructed in \cite{BFK1}. Here we need a
more general case.

\paragraph{Nested \textquotedblleft off-shell\textquotedblright\ Bethe ansatz
for particles of rank $1$ and $N-1:$}

We consider a state with $n$ particles of rank $1$ and $\bar{n}$ particles of
rank $N-1$ and write the off-shell Bethe ansatz co-vector valued function as
\begin{equation}
\fbox{$\rule[-0.2in]{0in}{0.5in}\displaystyle~K_{\underline{\alpha}%
\underline{(\rho)}}(\underline{\theta},\underline{\omega})=\int_{\mathcal{C}%
_{\underline{\theta}}\underline{\omega}}dz_{1}\cdots\int_{\mathcal{C}%
_{\underline{\theta}\underline{\omega}}}dz_{m}\,k(\underline{\theta
},\underline{\omega},\underline{z})\,\tilde{\Psi}_{\underline{\alpha
}\underline{(\rho)}}(\underline{\theta},\underline{\omega},\underline{z})$~}
\label{2.14}%
\end{equation}
where $\underline{\alpha}=(\alpha_{1},\dots,\alpha_{n}),~\underline{(\rho
)}=((\rho_{1}),\dots,(\rho_{\bar{n}}))=((\rho_{11},\dots,\rho_{1N-1}%
),\dots,(\rho_{\bar{n}1},\dots,\rho_{\bar{n}N-1}))$, $\underline{\theta
}=\left(  \theta_{1},\dots,\theta_{n}\right)  ,~\underline{\omega}=\left(
\omega_{1},\dots,\omega_{\bar{n}}\right)  $ and $\underline{z}=\left(
z_{1},\dots,z_{m}\right)  $. This ansatz transforms the complicated matrix
equations to simple equations for the scalar function $k(\underline{\theta
},\underline{\omega},\underline{z})$ (see \cite{BFK1} and below). The
integration contour $\mathcal{C}_{\underline{\theta}\underline{\omega}}$ can
in general be characterized as follows: there is a finite number of complex
numbers $a_{i}(\underline{\theta}),\,b_{j}(\underline{\theta})$ such that the
positions of all poles of the integrand are of the form%
\begin{equation}%
\begin{array}
[c]{l}%
(1):a_{i}(\underline{\theta})+2\pi ik\,,~~k\in\mathbb{N}\\
(2):b_{j}(\underline{\theta})-2\pi il\,,~~l\in\mathbb{N}%
\end{array}
\label{2.16}%
\end{equation}
and $\mathcal{C}_{\underline{\theta}\underline{\omega}}$ runs from $-\infty$
to $+\infty$ such that all poles (1) are above and all poles (2) are below the
contour. This contour is just the same as the one used for the definition of
Meijer's G-function. It will turn out that for the examples considered below
the form factors can be expressed in terms of Meijer's G-functions.

The state $\tilde{\Psi}_{\underline{\alpha}\underline{(\rho)}}$ in
(\ref{2.14}) is a linear combination of the basic Bethe ansatz co-vectors
\begin{equation}
\tilde{\Psi}_{\underline{\alpha}\underline{(\rho)}}(\underline{\theta
},\underline{\omega},\underline{z})=L_{\underline{\beta}\underline{(\sigma)}%
}(\underline{z},\underline{\omega})\tilde{\Phi}_{\underline{\alpha}%
\underline{(\rho)}}^{\underline{\beta}\underline{(\sigma)}}(\underline{\theta
},\underline{\omega},\underline{z})\,,~~\text{with }1<\beta_{i},~\sigma
_{1j}=1\,. \label{2.18}%
\end{equation}
As usual in the context of the algebraic Bethe ansatz \cite{FST,TF} the basic
Bethe ansatz co-vectors are obtained from the monodromy matrix
\begin{align*}
\tilde{T}_{1\dots n,\bar{1}\dots\bar{n},0}(\underline{\theta},\underline
{\omega},\theta_{0})  &  =\tilde{S}_{10}(\theta_{1}-\theta_{0})\,\cdots
\tilde{S}_{n0}(\theta_{n}-\theta_{0})\tilde{S}_{\bar{1}0}(\omega_{1}%
-\theta_{0})\,\cdots\tilde{S}_{\bar{n}0}(\omega_{\bar{n}}-\theta_{0})\\
&  \equiv\left(
\begin{array}
[c]{cc}%
\tilde{A}_{1\dots n,\bar{1}\dots\bar{n}}(\underline{\theta},\underline{\omega
},z) & \tilde{B}_{1\dots n,\bar{1}\dots\bar{n},\beta}(\underline{\theta
},\underline{\omega},z)\\
\tilde{C}_{1\dots n,\bar{1}\dots\bar{n}}^{\beta}(\underline{\theta}%
,\underline{\omega},z) & \tilde{D}_{1\dots n,\bar{1}\dots\bar{n},\beta}%
^{\beta^{\prime}}(\underline{\theta},\underline{\omega},z)
\end{array}
\right)  ,~~2\leq\beta,\beta^{\prime}\leq N\,.
\end{align*}
where the S-matrices $\tilde{S}_{i0}$ and $\tilde{S}_{\bar{\imath}0}$ are
given by (\ref{2.10}) and (\ref{2.12}). As usual the Yang-Baxter algebra
relation for the S-matrix yields the typical $TTS$-relation which implies the
basic algebraic properties of the sub-matrices $A,B,C,D$.

Here not only one reference co-vector exists. The space of reference
co-vectors, defined as usual by%
\[
\Omega^{\underline{(\sigma)}}\tilde{B}_{\beta}=0\,,
\]
is $(N-1)^{\bar{n}}$ dimensional and is spanned by the co-vectors for all
$\underline{(\sigma)}=((\sigma_{11},\dots,\sigma_{1N-1}),$$\allowbreak\dots
,$$(\sigma_{\bar{n}1},\dots,\sigma_{\bar{n}N-1}))$ with $\sigma_{i1}%
=1<\sigma_{i2}<\dots<\sigma_{iN-1}\leq N.$ They are eigenstates of $\tilde{A}$
and $\tilde{D}_{\beta}^{\beta^{\prime}}$%
\[%
\begin{array}
[c]{rcl}%
\Omega\,^{\underline{(\sigma)}}\tilde{A}(\underline{\theta},\underline{\omega
},z) & = & \Omega^{\underline{(\sigma)}}\\
\Omega^{\underline{(\sigma)}}\tilde{D}_{\beta}^{\beta^{\prime}}(\underline
{\theta},\underline{\omega},z) & = & \delta_{\beta}^{\beta^{\prime}}%
\prod\limits_{i=1}^{n}\tilde{b}(\theta_{i}-z)\Omega^{\underline{(\sigma)}}\,.
\end{array}
\]
where the indices $1\dots n,\bar{1}\dots\bar{n}$ are suppressed. The basic
Bethe ansatz co-vectors in (\ref{2.18}) are defined as%
\begin{equation}
\tilde{\Phi}_{\underline{\alpha}\underline{(\rho)}}^{\underline{\beta
}\underline{(\sigma)}}(\underline{\theta},\underline{\omega},\underline
{z})=\left(  \Omega^{\underline{(\sigma)}}\tilde{C}^{\beta_{m}}(\underline
{\theta},\underline{\omega},z_{m})\cdots\tilde{C}^{\beta_{1}}(\underline
{\theta},\underline{\omega},z_{1})\right)  _{\underline{\alpha}\underline
{(\rho)}} \label{2.20}%
\end{equation}
where $1<\beta_{i}\leq N$.

The technique of the \textbf{`nested Bethe ansatz'} means that for the
coefficients $L_{\underline{\beta}\underline{(\sigma)}}(\underline
{z},\underline{\omega})$ in (\ref{2.18}) one makes the analogous construction
as for $K_{\underline{\alpha}\underline{(\rho)}}(\underline{\theta}%
,\underline{\omega})$ where now the indices $\underline{\beta},\underline
{(\sigma)}$ take only the values $2\leq\beta_{i}\leq N$ and $\sigma
_{i1}=1<\sigma_{i2}<\dots<\sigma_{iN-1}\leq N$. This nesting is repeated until
the space of the coefficients becomes one dimensional. It is well known (see
\cite{BKZ2}) that the `off-shell' Bethe ansatz states are highest weight
states if they satisfy a certain matrix difference equation. If there are only
$n$ particles of rank $1$, then the $SU(N)$ weights are
\begin{equation}
w=\left(  n-n_{1},n_{1}-n_{2},\dots,n_{N-2}-n_{N-1},n_{N-1}\right)
\label{2.22}%
\end{equation}
where $n_{1}=m,n_{2},\dots$ are the numbers of $C$ operators in the various
levels of the nesting. If in addition there are $\bar{n}$ particles of rank
$N-1$ the $SU(N)$ weights are
\begin{equation}
w=\left(  n-n_{1},n_{1}-n_{2},\dots,n_{N-2}-n_{N-1},n_{N-1}-\bar{n}\right)
+\bar{n}(1,\dots,1) \label{2.24}%
\end{equation}
because $N-1$ particles of rank 1 yield a bound state of rank $N-1$ and at the
$l^{th}$ level the number of $C$ operators is reduced by $N-l-1$ (see appendix
\ref{sb}).

\subsection{Minimal form factors and $\phi$-function}

To construct the form factors we need the form factor functions $F\left(
\theta\right)  $,~$G\left(  \theta\right)  $ and the function $\phi(\theta)$.
The form factor functions $F\left(  \theta\right)  $ and $G\left(
\theta\right)  $ for two particles of rank 1 and for one particle of rank 1
and one of rank $N-1$, respectively are
\begin{align}
F\left(  \theta\right)   &  =c\exp\int\limits_{0}^{\infty}\frac{dt}{t\sinh
^{2}t}e^{\frac{t}{N}}\sinh t\left(  1-\frac{1}{N}\right)  \left(  1-\cosh
t\left(  1-\frac{\theta}{i\pi}\right)  \right) \label{2.26}\\
G(\theta)  &  =c^{\prime}\exp\int\limits_{0}^{\infty}\frac{dt}{t\sinh^{2}%
t}e^{\frac{t}{N}}\sinh\frac{t}{N}\left(  1-\cosh t\left(  1-\frac{\theta}%
{i\pi}\right)  \right)  \,. \label{2.28}%
\end{align}
They are the minimal solutions of the equations
\begin{gather*}
F(\theta)=F(-\theta)a(\theta)\,,~~F(i\pi-\theta)=F(i\pi+\theta)\\
G(\theta)=-G(-\theta)b(\pi i-\theta)\,,~~G(i\pi-\theta)=G(i\pi+\theta)
\end{gather*}
where $a(\theta)$ and $b(\pi i-\theta)$ are the highest weight amplitudes of
the corresponding channels of (\ref{2.2}) and (\ref{2.4}). The $\phi
$-function
\begin{equation}
\tilde{\phi}(\theta)=\frac{1}{F(-\theta)G(i\pi+\theta)}=\Gamma\left(
\frac{-\theta}{2\pi i}\right)  \Gamma\left(  1-\frac{1}{N}+\frac{\theta}{2\pi
i}\right)  \label{2.30}%
\end{equation}
is a solution of%
\begin{equation}
\prod_{k=0}^{N-2}\tilde{\phi}\left(  -\theta-ki\eta\right)  \prod_{k=0}%
^{N-1}F\left(  \theta+ki\eta\right)  =1 \label{2.32}%
\end{equation}
which follows from the assumption that the antiparticle of a fundamental
particle is a bound state of $N-1$ of them (see below and \cite{BFK1}). The
constants $c$ and $c^{\prime}$ in (\ref{2.26}) and (\ref{2.28}) follow from
(\ref{2.30}) and (\ref{2.32}).

\subsection{Chiral Gross-Neveu model}

Swieca et al.~\cite{KKS} wrote the fermion fields $\psi_{i}(x)$ in the
Lagrangian (\ref{1.2}) in bosonic form. In order to extract the real particle
content of the theory, they introduced in addition the \textquotedblleft
physical\textquotedblright\ fields%
\begin{align*}
\hat{\psi}_{i}(x)  &  =\mathcal{K}_{i}\left(  \frac{m}{2\pi}\right)
^{1/2}e^{(\pi/4)\gamma^{5}}\exp\left\{  -i\sqrt{\pi}\left(  \gamma^{5}\phi
_{i}(x)+\int_{x}^{\infty}dy^{1}\dot{\phi}_{i}(y)\right)  \right\} \\
\phi_{i}(x)  &  =\left(  1-\frac{1}{N}\right)  \varphi_{i}(x)-\frac{1}{N}%
\sum_{j\neq i}\varphi_{j}(x)
\end{align*}
where $\varphi_{i}(x)$ are free canonical zero-mass fields and $\mathcal{K}%
_{i}$ are Klein factors satisfying%
\[
\mathcal{K}_{i}\hat{\psi}_{j}(x)=\left\{
\begin{array}
[c]{lll}%
\hat{\psi}_{j}(x)\mathcal{K}_{i} & \text{for} & i=j\\
-\hat{\psi}_{j}(x)\mathcal{K}_{i} & \text{for} & i\neq j~.
\end{array}
\right.
\]
The fields $\hat{\psi}$ satisfy (with a suitable normal product prescription
$\mathcal{N}$)
\begin{equation}
\hat{\psi}_{i}^{\dagger}=\mathcal{K}\frac{1}{(N-1)!}\sum_{\underline{j}%
}\epsilon_{ij_{1}\dots j_{N-1}}\mathcal{N}\hat{\psi}_{j_{1}}\dots\hat{\psi
}_{j_{N-1}}\,,~\mathcal{K}=%
%TCIMACRO{\dprod \limits_{j=1}^{N}}%
%BeginExpansion
{\displaystyle\prod\limits_{j=1}^{N}}
%EndExpansion
\mathcal{K}_{j} \label{2.34}%
\end{equation}
such that%
\begin{equation}
\mathcal{K}\hat{\psi}_{j}(x)=(-1)^{N-1}\hat{\psi}_{j}(x)\mathcal{K}\,.
\label{2.36}%
\end{equation}
Equation (\ref{2.34}) means that antiparticles should be identified with bound
state of $N-1$ particles and the creation operators of the antiparticle
$\hat{b}_{\alpha}^{\dagger}$ and of the bound state $\hat{a}_{(\varrho
)}^{\dagger}$ are related by%
\begin{align*}
\hat{b}_{\alpha}^{\dagger}  &  =\mathcal{K}\epsilon_{\alpha(\varrho)}\hat
{a}_{(\varrho)}^{\dagger}\\
\mathcal{K}\hat{a}_{\alpha}^{\dagger}\mathcal{K}  &  =(-1)^{N-1}\hat
{a}_{\alpha}^{\dagger}\,.
\end{align*}
The \textquotedblleft physical\textquotedblright\ fields satisfy the anyonic
commutation relations%
\[
\hat{\psi}_{i}(x)\hat{\psi}_{i}(y)=\hat{\psi}_{i}(y)\hat{\psi}_{i}(x)e^{2\pi
is\epsilon(x^{1}-y^{1})},~\text{for }(x-y)^{2}<0
\]
with \textquotedblleft spin\textquotedblright\ $s=\tfrac{1}{2}\left(
1-1/N\right)  $. This implies that the \textquotedblleft
physical\textquotedblright\ S-matrix is related to the one of (\ref{2.2}) by
\cite{KKS,KT1}%
\[
S_{\alpha\beta}^{\delta\gamma}(\theta_{12})=e^{2\pi is\epsilon(\theta_{12}%
)}\hat{S}_{\alpha\beta}^{\delta\gamma}(\theta_{12})\,.
\]
As a consequence the abnormal crossing relation (\ref{2.4}) transforms to a
normal one. The bound state S-matrix satisfies
\begin{equation}
S_{(\rho)\beta}^{\delta(\sigma)}(\theta)=(-1)^{N-1}\mathbf{C}_{(\rho)\gamma
}S_{\beta\alpha}^{\gamma\delta}(\pi i-\theta)\mathbf{C}^{\alpha(\sigma)}.
\label{2.38}%
\end{equation}
Therefore the physical crossing relation is%
\[
\hat{S}_{\bar{\alpha}\beta}^{\delta\bar{\gamma}}(\theta)=\mathbf{\hat{C}%
}_{\bar{\alpha}\alpha^{\prime}}\hat{S}_{\beta\gamma^{\prime}}^{\alpha^{\prime
}\delta}(\pi i-\theta)\mathbf{\hat{C}}^{\gamma^{\prime}\bar{\gamma}}%
\]
with $\mathbf{\hat{C}}_{\bar{\alpha}\alpha^{\prime}}=\delta_{\alpha
\alpha^{\prime}},\ \mathbf{\hat{C}}^{\gamma^{\prime}\bar{\gamma}}%
=\delta^{\gamma^{\prime}\gamma}$.

\section{Form factors}

\label{s3}For a state of $n$ particles of rank $r_{1},\dots,r_{n}$ with
rapidities $\underline{\theta}=(\theta_{1},\dots,\theta_{n})$ and a local
operator $\mathcal{O}(x)$ we define the associated form factor functions
$F_{\underline{\alpha}}^{\mathcal{O}}(\underline{\theta})$ by
\[
\langle\,0\,|\,\mathcal{O}(x)\,|\,\theta_{1},\dots,\theta_{n}\,\rangle
_{\underline{\alpha}}^{in}=e^{-ix(p_{1}+\cdots+p_{n})}F_{\underline{\alpha}%
}^{\mathcal{O}}(\underline{\theta})~,~~\text{for}~\theta_{1}>\dots>\theta
_{n}.
\]
where again $\underline{\alpha}=((\alpha_{11},\dots,\alpha_{1r_{1}}%
),\dots,(\alpha_{n1},\dots,\alpha_{nr_{n}}))$. For all other arrangements of
the rapidities the functions $F_{\underline{\alpha}}^{\mathcal{O}}%
(\underline{\theta})$ are given by analytic continuation. The co-vector valued
function $F_{\underline{\alpha}}^{\mathcal{O}}(\underline{\theta})$ satisfies
the form factor equations (i) -- (v) (see \cite{KW,Sm,BFKZ,BK,BFK1}) and can
be written as \cite{KW}
\begin{equation}
F_{\underline{\alpha}}^{\mathcal{O}}(\underline{\theta})=K_{\underline{\alpha
}}^{\mathcal{O}}(\underline{\theta})\prod_{1\leq i<j\leq n}F_{r_{i}r_{j}%
}(\theta_{ij}) \label{3.2}%
\end{equation}
where $F_{r_{i}r_{j}}(\theta)$ are the minimal form factor functions. For
particles of rank $1$ and $N-1$ they are given by $F_{11}(\theta
)=F_{N-1N-1}(\theta)=F(\theta)$ and $F_{N-11}(\theta)=F_{1N-1}(\theta
)=G(\theta)$ of (\ref{2.26}) and (\ref{2.28}), respectively. In \cite{BFK1}
the form factors of the fundamental particles of rank $1$ have been
constructed. We shortly recall the results. All other form factors can be
obtained from these by applying the bound state fusion procedure which is
given by the form factor equation (iv) (see e.g. \cite{BFK1}).

\paragraph{Form factors for particles of rank 1:}

The K-function in (\ref{3.2}) is given by the nested \textquotedblleft
off-shell\textquotedblright\ Bethe ansatz (\ref{2.14}) for the special case
$\bar{n}=0$ and a special choice of the scalar function $\,k(\underline
{\theta},\underline{z})$ such that the form factor equations (i) -- (v) are
satisfied%
\begin{equation}
\fbox{$\rule[-0.2in]{0in}{0.5in}\displaystyle~K_{\underline{\alpha}%
}^{\mathcal{O}}(\underline{\theta})=\frac{N_{n}}{m!}\int_{\mathcal{C}%
_{\underline{\theta}}}dz_{1}\cdots\int_{\mathcal{C}_{\underline{\theta}}%
}dz_{m}\,\tilde{h}(\underline{\theta},\underline{z})\,p^{\mathcal{O}%
}(\underline{\theta},\underline{z})\,\tilde{\Psi}_{\underline{\alpha}%
}(\underline{\theta},\underline{z})$~}\label{3.4}%
\end{equation}
with
\begin{align}
\tilde{h}(\underline{\theta},\underline{z}) &  =\prod_{i=1}^{n}\prod_{j=1}%
^{m}\tilde{\phi}(\theta_{i}-z_{j})\prod_{1\leq i<j\leq m}\tau(z_{i}%
-z_{j})\label{3.6}\\
\tau(z) &  =\frac{1}{\tilde{\phi}(z)\tilde{\phi}(-z)}\,.\nonumber
\end{align}
The integration contour $\mathcal{C}_{\underline{\theta}}$ is defined by
(\ref{2.16}). The dependence on the operator $\mathcal{O}$ enters only through
the p-function $p^{\mathcal{O}}(\underline{\theta},\underline{z})$ which has
to satisfy simple equations (see \cite{BK2,BK04,BFK,BFK1}). The K-function is
in general a linear combination of the \emph{fundamental building blocks}
\cite{BK2,BK04,BFK} given by (\ref{3.4}). Here we consider only these cases
where the sum consists only of one term.

\subparagraph{The p-function:}

The co-vector valued function $\tilde{\Psi}_{\underline{\alpha}}%
(\underline{\theta},\underline{z})$ is expressed as in (\ref{2.18}) for
$\bar{n}=0$ by an $L_{\underline{\beta}}(\underline{z})$ for which the nesting
procedure is applied. The final form is (up to a constant)
\begin{align}
F_{\underline{\alpha}}^{\mathcal{O}}(\underline{\theta})  &  =\prod
F(\theta_{ij})\int d\underline{z}^{(1)}\,\dots\int d\underline{z}%
^{(N-1)}\tilde{h}(\underline{\theta},\underline{\underline{z}}%
)\,p^{\mathcal{O}}(\underline{\theta},\underline{\underline{z}})\,\tilde{\Phi
}_{\underline{\alpha}}(\underline{\theta},\underline{\underline{z}%
})\label{3.8}\\
\tilde{h}(\underline{\theta},\underline{\underline{z}})  &  =\tilde
{h}(\underline{\theta},\underline{z}^{(1)})\cdots\tilde{h}(\underline
{z}^{(N-2)},\underline{z}^{(N-1)})\,.\nonumber
\end{align}
where $\underline{\underline{z}}=\underline{z}^{(1)},\dots,\underline
{z}^{(N-1)}$. The Bethe state $\tilde{\Phi}_{\underline{\alpha}}%
(\underline{\theta},\underline{\underline{z}})$ is obtained by the nesting
procedure (see (\ref{2.18}), (\ref{2.20}) and \cite{BFK1})%

\[
\tilde{\Phi}_{\underline{\alpha}}(\underline{\theta},\underline{\underline{z}%
})=\tilde{\Phi}_{\underline{\varsigma}}^{(N-1)}(\underline{z}^{(N-2)}%
,\underline{z}^{(N-1)})\dots\tilde{\Phi}_{\underline{\beta}}^{(2),\underline
{\gamma}}(\underline{z}^{(1)},\underline{z}^{(2)})\tilde{\Phi}_{\underline
{\alpha}}^{\underline{\beta}}(\underline{\theta},\underline{z}^{(1)})\,.
\]
In general the p-function (see \cite{BFK1}) depends on the rapidities
$\underline{\theta}$ and all integration variables $\underline{z}^{(l)}$. Let
the operator $\mathcal{O}(x)$ transform as a highest weight $SU(N)$
representation with highest weight vector
\[
w^{\mathcal{O}}=\left(  w_{1}^{\mathcal{O}},\dots,w_{N}^{\mathcal{O}}\right)
.
\]
Because of $SU(N)$ invariance the weight vector of the co-vector
$F_{\underline{\alpha}}^{\mathcal{O}}(\underline{\theta})$ is then%
\begin{align}
w  &  =\left(  w_{1}^{\mathcal{O}},\dots,w_{N}^{\mathcal{O}}\right)  +L\left(
1,\dots,1\right) \label{3.10}\\
&  =\left(  n-n_{1},n_{1}-n_{2},\dots,n_{N-2}-n_{N-1},n_{N-1}\right) \nonumber
\end{align}
where (\ref{2.22}) and the fact, that the weight vector $\left(
1,\dots,1\right)  $ correspond to the vacuum sector, has been used. In
\cite{BFK1} was shown that the p-function has to satisfy a set of equations in
order that the form factor (\ref{3.8}) satisfies the form factor equations. In
particular to guarantee the transformation properties of the operator the
following periodicity relations have to be valid
\begin{equation}%
\begin{array}
[c]{ccc}%
p^{\mathcal{O}}(\underline{\theta},\dots,\underline{z}^{(l)},\dots) & = &
\tilde{\sigma}_{1}^{\mathcal{O}}p^{\mathcal{O}}(\dots,\theta_{i}+2\pi
i,\dots,\underline{z}^{(l)},\dots)\\
& = & (-1)^{w_{l}^{\mathcal{O}}+w_{l+1}^{\mathcal{O}}}p^{\mathcal{O}%
}(\underline{\theta},\dots,z_{i}^{(l)}+2\pi i,\dots)
\end{array}
\label{3.12}%
\end{equation}
where%
\begin{align}
\tilde{\sigma}_{1}^{\mathcal{O}}  &  =\sigma_{1}^{\mathcal{O}}%
(-1)^{(N-1)\left[  \sum_{i=1}^{N}w_{i}^{\mathcal{O}}/N\right]  -\sum_{i=2}%
^{N}w_{i}^{\mathcal{O}}}\nonumber\\
\sigma_{1}^{\mathcal{O}}  &  =e^{i\pi(1-1/N)Q^{\mathcal{O}}}\,. \label{3.14}%
\end{align}
The charge of the operator $\mathcal{O}$ is defined by $Q^{\mathcal{O}%
}=n\operatorname{mod}N$ and $\sigma_{1}^{\mathcal{O}}$ is the statistics
factor of $\mathcal{O}$ with respect to the fundamental particle of rank $1$.
The sign factors $\tilde{\sigma}_{1}^{\mathcal{O}}/\sigma_{1}^{\mathcal{O}%
}=\pm1$ and $(-1)^{w_{l}^{\mathcal{O}}+w_{l+1}^{\mathcal{O}}}=\pm1$ in
(\ref{3.12}) follow \cite{BFK1} from the sign $(-1)^{(N-1)}$ in the unusual
crossing relation (\ref{2.38}) (related to the Klein factors of (\ref{2.36})).

\subsection{General form factors of particles of rank $1$ and $N-1$:}

Using the bound state procedure (see appendix \ref{sb}) which means taking
residues of (\ref{3.2}) or (\ref{3.4}) one derives the form factors and
K-functions for $n$ particles of rank $1$ with rapidities $\underline{\theta}$
and $\bar{n}$ particles of rank $N-1$ with rapidities $\underline{\omega}$. As
usual we split off the minimal part
\begin{equation}
F_{\underline{\alpha}\underline{(\rho)}}^{\mathcal{O}}(\underline{\theta
},\underline{\omega})=K_{\underline{\alpha}\underline{(\rho)}}(\underline
{\theta},\underline{\omega})\prod_{1\leq i<j\leq n}F(\theta_{ij})\prod
_{i=1}^{n}\prod_{j=1}^{\bar{n}}G(\theta_{i}-\omega_{j})\prod_{1\leq
i<j\leq\bar{n}}F(\omega_{ij})\,. \label{3.16}%
\end{equation}
The K-function is given by a nested `off-shell' Bethe ansatz (\ref{2.14})
\begin{equation}
\fbox{$\rule[-0.2in]{0in}{0.5in}\displaystyle~K_{\underline{\alpha}%
\underline{(\rho)}}^{\mathcal{O}}(\underline{\theta},\underline{\omega}%
)=\frac{N_{n\bar{n}}}{m!}\int_{\mathcal{C}_{\underline{\theta}}\underline
{\omega}}dz_{1}\cdots\int_{\mathcal{C}_{\underline{\theta}\underline{\omega}}%
}dz_{m}\,\tilde{h}(\underline{\theta},\underline{z})\,p^{\mathcal{O}%
}(\underline{\theta},\underline{\omega},\underline{z})\,\tilde{\Psi
}_{\underline{\alpha}\underline{(\rho)}}(\underline{\theta},\underline{\omega
},\underline{z})$~} \label{3.18}%
\end{equation}
where $\tilde{h}(\underline{\theta},\underline{z})$ is the scalar function
(\ref{3.6}). Note that this h-function does not depend on $\underline{\omega}%
$. For the $SU(N)$ S-matrix the function $\tilde{\phi}(\theta)$ is given by
(\ref{2.30}). The integration contour $\mathcal{C}_{\underline{\theta
}\underline{\omega}}$ (see Fig. \ref{f2}) has been defined in the context of
(\ref{2.14}). \begin{figure}[tbh]%
\[
\unitlength4mm\begin{picture}(27,13)
\thicklines
\put(1,0){
\put(0,0){$\bullet$}
\put(-.3,1){$\omega-i\pi-2\pi i\frac1N$}
\put(0,6){$\bullet$}
\put(-.3,5){$\omega+i\pi-2\pi i\frac1N$}
\put(0,7.2){$\bullet$}
\put(-.3,8.2){$\omega +i\pi$}
\put(0,12.2){$\bullet$}
\put(-.3,11.4){$\omega+3\pi$}
}
\put(12,0){
\put(0,0){$\bullet~\theta_2-2\pi i$}
\put(0,4.8){$\bullet$}
\put(.19,3.8){$\theta_2-2\pi i\frac1N$}
\put(0,6){$\bullet~\theta_2$}
\put(0,11){$\bullet\,\theta_2+2\pi i(1-\frac1N)$}
}
\put(20,1){
\put(0,0){$\bullet~\theta_1-2\pi i$}
\put(0,4.8){$\bullet$}
\put(.19,3.8){$\theta_1-2\pi i\frac1N$}
\put(0,6){$\bullet~\theta_1$}
\put(0,11){$\bullet\,\theta_1+2\pi i(1-\frac1N)$}
}
\put(10,3){\vector(1,0){0}}
\put(15,5.6){\vector(-1,0){0}}
\put(16,8.6){\vector(1,0){0}}
\put(-1,4.5){\oval(18,4.7)[tr]}
\put(24,4.5){\oval(32,3)[bl]}
\put(24,4.8){\oval(3,3.6)[r]}
\put(24,6.1){\oval(14,1)[tl]}
\put(12,6.1){\oval(10,1)[br]}
\put(12,7.1){\oval(4,3)[l]}
\put(12,8.6){\line(1,0){14}}
\end{picture}
\]
\caption{The integration contour $\mathcal{C}_{\theta_{1}\theta_{2}\omega}$
for two particles and one bound state.}%
\label{f2}%
\end{figure}
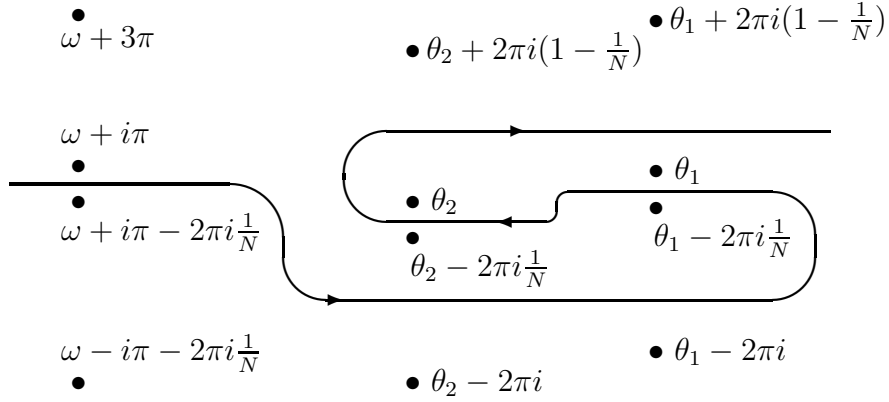

\subparagraph{Nesting:}

The state $\tilde{\Psi}_{\underline{\alpha}\underline{(\rho)}}$ in
(\ref{3.18}) is a linear combination of the basic Bethe ansatz co-vectors
(\ref{2.20})%
\[
\tilde{\Psi}_{\underline{\alpha}\underline{(\rho)}}(\underline{\theta
},\underline{\omega},\underline{z})=L_{\underline{\beta}\underline{(\sigma)}%
}(\underline{z},\underline{\omega})\tilde{\Phi}_{\underline{\alpha}%
\underline{(\rho)}}^{\underline{\beta}\underline{(\sigma)}}(\underline{\theta
},\underline{\omega},\underline{z})\,,~~\text{with }1<\beta_{i},~\sigma
_{1j}=1
\]
where $L_{\underline{\beta}(\sigma)}(\underline{z},\underline{\omega})$
satisfies again a representation like (\ref{3.18}). This nesting is iterated
until all $\beta_{i}=N$ and all $(\sigma)_{i}=(1,2,\dots,N-1)$. Only for the
highest level Bethe ansatz the h-function depends on $\underline{\omega}$. The
final result is%
\begin{align}
K_{\underline{\alpha}(\rho)}^{\mathcal{O}}(\underline{\theta},\omega)  &
=\int d\underline{z}^{(1)}\dots\int d\underline{z}^{(N-1)}\tilde
{h}\,(\underline{\theta},\underline{\omega},\underline{\underline{z}%
})p^{\mathcal{O}}(\underline{\theta},\underline{\omega},\underline
{\underline{z}})\tilde{\Phi}_{\underline{\alpha}(\rho)}(\underline{\theta
},\underline{\omega},\underline{\underline{z}})\,\label{3.20}\\
\tilde{h}\,(\underline{\theta},\underline{\omega},\underline{\underline{z}})
&  =\prod_{l=0}^{N-2}\tilde{h}(\underline{z}^{(l)},\underline{z}^{(l+1)}%
)\prod_{i=1}^{\bar{n}}\prod_{j=1}^{n_{N-1}}\tilde{\chi}(\omega_{i}%
-z_{j}^{(N-1)})\nonumber\\
\tilde{\chi}(\omega)  &  =\Gamma\left(  \frac{1}{2}+\frac{\omega}{2\pi
i}\right)  \Gamma\left(  \frac{1}{2}-\frac{1}{N}-\frac{\omega}{2\pi i}\right)
\,.\nonumber
\end{align}
The complete Bethe ansatz state is%
\[
\tilde{\Phi}_{\underline{\alpha}\underline{(\rho)}}(\underline{\theta
},\underline{\omega},\underline{\underline{z}})=\tilde{\Phi}_{\underline
{\varsigma}\underline{(\lambda)}}^{(N-2)\underline{(\eta)}}(\underline
{z}^{(N-2)},\underline{\omega},\underline{z}^{(N-1)})\dots\tilde{\Phi
}_{\underline{\beta}\underline{(\sigma)}}^{(1)\underline{\gamma}%
\underline{(\kappa)}}(\underline{z}^{(1)},\underline{\omega},\underline
{z}^{(2)})\tilde{\Phi}_{\underline{\alpha}\underline{(\rho)}}^{\underline
{\beta}\underline{(\sigma)}}(\underline{\theta},\underline{\omega}%
,\underline{z}^{(1)})
\]
where $\underline{(\eta)}$ denotes $\bar{n}$ highest weight bound states
$(\eta_{i1},\dots,\eta_{iN-1})=(1,2,\dots,N-1)$. The p-functions in
(\ref{3.18}) and (\ref{3.20}) are obtained again by the bound state procedure
from a solution of (\ref{3.12}) for $\bar{n}=0$. In particular for $\bar{n}=1$
(with the replacements in (\ref{3.12}) $\underline{\theta}\rightarrow
\underline{\theta},\underline{\varphi}$ and $\underline{z}^{(l)}%
\rightarrow\underline{z}^{(l)},\underline{y}^{(l)}$ where $\underline{\varphi
}=\left(  \varphi_{1},\dots,\varphi_{N-1}\right)  ,~\underline{y}%
^{(l)}=\left(  y_{1},\dots,y_{N-1-l}\right)  ,~l=1,\dots,N-2$)
\[
p^{\mathcal{O}}(\underline{\theta},\omega,\underline{\underline{z}%
})=p^{\mathcal{O}}(\underline{\theta},\underline{\varphi},\underline{z}%
^{(1)},\underline{y}^{(1)},\dots,\underline{z}^{(N-1)},\underline{y}%
^{(N-1)})\,.
\]
Here $y_{i}^{(l)}=\varphi_{i}^{(l)},~i=1,\dots,N-1-l$ and $\varphi_{k}%
=\omega+ki\eta-i\pi.$ The proofs of the statements of this subsection and more
details can be found in appendix \ref{sb}.

\subsection{Examples}

To illustrate our general results we present some simple examples. In
addition, we also derive the $1/N$ expansion of exact form factors for the
purpose of later comparison with the $1/N$-perturbation theory of the chiral
$SU(N)$ Gross-Neveu model.

\paragraph{The energy momentum tensor:}

For the local operator $\mathcal{O}(x)=T^{\rho\sigma}(x)$ (where $\rho
,\sigma=\pm$ denote the light cone components) the p-function for $n$
particles of rank 1 (as for the sine-Gordon model in \cite{BK})
\begin{equation}
p^{T^{\rho\sigma}}(\underline{\theta},\underline{z})=\sum\limits_{i=1}%
^{n}e^{\rho\theta_{i}}\sum\limits_{i=1}^{m}e^{\sigma z_{i}} \label{3.22}%
\end{equation}
satisfies the equations (\ref{3.12}) with $w^{T}=\left(  0,0,\dots,0\right)
$. For the $n=N$ particle form factor the weight vector is $w=\left(
1,1,\dots,1,1\right)  $. Due to (\ref{2.22}) there are $n_{l}=N-l$
integrations in the $l$-th level of the off-shell Bethe ansatz.

We calculate the form factor of the particle $\alpha$ and the bound state
$(\lambda)=(\lambda_{1},\dots,\lambda_{N-1})$ of $N-1$ particles. We apply the
bound state formulae (\ref{3.16}) and (\ref{3.18}) for $n=\bar{n}=1$. Due to
(\ref{2.24}) there is just one integration in every level of the nested Bethe
ansatz ($l=1,\dots,N-1$)%

\begin{align}
F_{\alpha(\lambda)}^{T^{\rho\sigma}}(\theta,\omega)  &  =K_{\alpha(\lambda
)}^{T^{\rho\sigma}}(\theta,\omega)\,G(\theta-\omega)\nonumber\\
K_{\alpha(\lambda)}^{T^{\rho\sigma}}(\theta,\omega)  &  =N_{2}^{T^{\rho\sigma
}}\left(  e^{\rho\theta}+e^{\rho\omega}\right)  \int_{\mathcal{C}%
_{\underline{\theta}}}dz\tilde{\phi}(\theta-z)e^{\sigma z}\tilde{\Psi}%
_{\alpha(\lambda)}(\theta,\omega,z)\label{3.24}\\
\tilde{\Psi}_{\alpha(\lambda)}(\theta,\omega,z)  &  =L_{\beta(\mu)}%
^{(1)}(z,\omega)\tilde{\Phi}_{\alpha(\lambda)}^{\beta(\mu)}(\theta
,\omega,z)\nonumber
\end{align}
where $G(\theta)$ defined in (\ref{2.28}) is the minimal form factor function
of two particles of rank $1$ and $N-1$. The integration in every level of the
nested Bethe ansatz ($l=N-2,\dots,1$) can be solved iteratively
\begin{align}
L_{\beta(\mu)}^{(l)}(z,\omega)  &  =\epsilon_{\beta(\mu)}L^{(l)}%
(\omega-z)~~\text{with }\beta>l,~(\mu)=(1,2,\dots,l,\ast,\dots,\ast
)\nonumber\\
L^{(l)}(\omega-z)  &  =c_{l}\ \Gamma\left(  \frac{1}{2}+\frac{\omega-z}{2\pi
i}\right)  \Gamma\left(  -\frac{1}{2}+\frac{l}{N}-\frac{\omega-z}{2\pi
i}\right)  \,. \label{3.26}%
\end{align}
The remaining integral in(\ref{3.24}) may be performed (see appendix
\ref{sb}\ ) with the result\footnote{In \cite{BKZ1,BKZ2} this result has been
obtained using Jackson type integrals.}%
\begin{equation}
\langle\,0\,|\,T^{\rho\sigma}(0)\,|\,\theta,\omega\,\rangle_{\alpha(\lambda
)}^{in}=4m_{1}^{2}\epsilon_{\alpha(\lambda)}e^{\frac{1}{2}(\rho+\sigma)\left(
\theta+\omega+i\pi\right)  }\frac{\sinh\tfrac{1}{2}\left(  \theta-\omega
-i\pi\right)  }{\theta-\omega-i\pi}G(\theta-\omega)\,. \label{3.28}%
\end{equation}
Similar as in \cite{BK} one can prove the eigenvalue equation
\[
\left(  \int dxT^{\pm0}(x)-\sum_{i=1}^{n}p_{i}^{\pm}\right)  |\,\theta
_{1},\dots,\theta_{n}\rangle_{\underline{\alpha}}^{in}=0
\]
for arbitrary states.

\paragraph{The fields $\psi_{\alpha}(x)$:}

Because the Bethe ansatz yields highest weight states we obtain the matrix
elements of the spinor field $\psi(x)=\psi_{1}(x)$. The p-function for the
local operator $\psi^{(\pm)}(x)$ for $n$ particles of rank 1 (see also
\cite{BFKZ})%
\[
p^{\psi^{(\pm)}}(\underline{\theta},\underline{z})=\exp\pm\frac{1}{2}\left(
\sum\limits_{i=1}^{m}z_{i}-\left(  1-\frac{1}{N}\right)  \sum\limits_{i=1}%
^{n}\theta_{i}\right)  \,.
\]
satisfies the equations (\ref{3.12}) with $w^{\psi}=\left(  1,0,\dots
,0\right)  $. For example the 1-particle form factor is
\[
\langle\,0\,|\,\psi^{(\pm)}(0)\,|\,\theta\,\rangle_{\alpha}=\delta_{\alpha
1}\,e^{\mp\frac{1}{2}\left(  1-\frac{1}{N}\right)  \theta}\,.
\]
The last formula is consistent with the proposal of Swieca et
al.~\cite{KuS,KKS} that the statistics of the fundamental particles in the
chiral $SU(N)$ Gross-Neveu model should be $\sigma=\exp\left(  2\pi is\right)
$, where $s=\frac{1}{2}\left(  1-\frac{1}{N}\right)  $ is the spin (see also
(\ref{3.14})). For the $n=N+1$ particle form factor there are again
$n_{l}=N-l$ integrations in the $l$-th level of the off-shell Bethe ansatz and
the $SU(N)$ weights are $w=\left(  2,1,\dots,1,1\right)  $. Due to
(\ref{2.24}) there is again just one integration in every level of the nested
Bethe ansatz. Similar as above one obtains the two-particle and one-bound
state form factor%
\begin{align}
F_{\alpha\beta(\lambda)}^{\psi^{(\pm)}}(\underline{\theta})  &  =K_{\alpha
\beta(\lambda)}^{\psi^{(\pm)}}(\underline{\theta})F(\theta_{12})G(\theta
_{13})G(\theta_{23})\nonumber\\
K_{\alpha\beta(\lambda)}^{\psi^{(\pm)}}(\underline{\theta})  &  =N^{\psi
}e^{\mp\frac{1}{2}\left(  \left(  1-\frac{1}{N}\right)  \left(  \theta
_{1}+\theta_{2}\right)  +\frac{1}{N}\theta_{3}\right)  }\int_{\mathcal{C}%
_{\underline{\theta}}}dz\tilde{\phi}(\theta_{1}-z)\tilde{\phi}(\theta
_{2}-z)e^{\pm\frac{1}{2}z}\tilde{\Psi}_{\alpha\beta(\lambda)}(\underline
{\theta},z)\label{3.30}\\
\tilde{\Psi}_{\alpha\beta(\lambda)}(\underline{\theta},z)  &  =L_{\gamma(\mu
)}(z,\theta_{3})\tilde{\Phi}_{\alpha\beta(\lambda)}^{\gamma(\mu)}%
(\underline{\theta},z)\,,~~\text{with }1<\gamma,~\lambda_{1}=1\nonumber
\end{align}
where the function $L_{\gamma(\mu)}(z,\theta_{3})=\epsilon_{\gamma(\mu
)}L^{(1)}(\theta_{3}-z)$ is the same as in (\ref{3.26}) above. We were not
able to perform this integration, however, the result can be expressed in
terms of Meijer's G-functions%
\begin{align*}
K_{\alpha\beta(\lambda)}^{\psi_{\delta}^{(\pm)}}(\underline{\theta})  &
=\epsilon_{\alpha(\lambda)}\delta_{\beta}^{\delta}K_{1}^{\psi^{(\pm)}%
}(\underline{\theta})+\epsilon_{\beta(\lambda)}\delta_{\alpha}^{\delta}%
K_{2}^{\psi^{(\pm)}}(\underline{\theta})\\
K_{1}^{\psi^{(\pm)}}(\underline{\theta})  &  =N_{1}^{\psi}e^{\mp\frac{1}%
{2}\left(  \left(  1-\frac{1}{N}\right)  \left(  \theta_{1}+\theta_{2}\right)
+\frac{1}{N}\theta_{3}\right)  }\,G_{33}^{33}\left(  e^{\pm i\pi}\left\vert
\begin{array}
[c]{c}%
\frac{\theta_{1}}{2\pi i}+1,\frac{\theta_{2}}{2\pi i}+1,\frac{\theta_{3}}{2\pi
i}+\frac{3}{2}-\frac{1}{N}\\
\frac{\theta_{1}}{2\pi i}-\frac{1}{N},\frac{\theta_{2}}{2\pi i}-\frac{1}%
{N}+1,\frac{\theta_{3}}{2\pi i}+\frac{1}{2}%
\end{array}
\right.  \right)
\end{align*}
and $K_{2}^{\psi^{(\pm)}}$ is obtained by the form factor equation (i).

\subparagraph{1/N expansion of the exact form factor:}

We consider the connected part of the matrix element%
\[
^{\gamma}\langle\,\theta_{3}\,|\,\psi_{\delta}^{(\pm)}(x)\,|\,\theta
_{1},\theta_{2}\,\rangle_{\alpha\beta}^{in,conn.}=\mathbf{C}^{(\lambda)\gamma
}F_{\alpha\beta(\lambda)}^{\psi^{(\pm)}}(\theta_{1},\theta_{2},\theta_{3}%
-i\pi)\,.
\]
Instead of the field $\psi$ we consider the operator $\mathcal{O}_{\delta
}=\left(  -i\left(  i\gamma\partial-m\right)  \psi\right)  _{\delta}$%
\[
F_{\alpha\beta}^{\mathcal{O}_{\delta},\gamma}=F_{(1)}^{\mathcal{O}}%
\delta_{\alpha}^{\gamma}\delta_{\beta}^{\delta}-F_{(2)}^{\mathcal{O}}%
\delta_{\beta}^{\gamma}\delta_{\alpha}^{\delta}~,~~F_{(2)}^{\mathcal{O}%
}(\theta_{1},\theta_{2},\theta_{3})=F_{(1)}^{\mathcal{O}}(\theta_{2}%
,\theta_{1},\theta_{3})\,.
\]
For $N\rightarrow\infty$ we expand the minimal form factors%
\[
F\left(  \theta\right)  =\frac{-i}{\pi}\sinh\tfrac{1}{2}\theta
+O(1/N)~,~~G(\theta)=1+O(1/N)\,,
\]
perform the integration in (\ref{3.30}) and obtain (after a lengthy
calculation)%
\begin{equation}
F_{(1)}^{\mathcal{O}}=-\frac{2mi\pi}{N}\frac{\sinh\theta_{13}}{\theta_{13}%
}\left(  \frac{1}{\cosh\frac{1}{2}\theta_{13}}-\gamma^{5}\frac{1}{\sinh
\frac{1}{2}\theta_{13}}\right)  u(\theta_{2})+O(N^{-2})\,. \label{3.32}%
\end{equation}
We use the following conventions for the $\gamma$-matrices
\[
\gamma^{0}=\left(
\begin{array}
[c]{cc}%
0 & 1\\
1 & 0
\end{array}
\right)  ~,\quad\gamma^{1}=\left(
\begin{array}
[c]{cc}%
0 & 1\\
-1 & 0
\end{array}
\right)  ~,\quad\gamma^{5}=\gamma^{0}\gamma^{1}=\left(
\begin{array}
[c]{cc}%
-1 & 0\\
0 & 1
\end{array}
\right)
\]
and for the spinors%
\begin{equation}
u(p)=\sqrt{m}\left(
\begin{array}
[c]{c}%
e^{-\theta/2}\\
e^{\theta/2}%
\end{array}
\right)  ~,\quad v(p)=\sqrt{m}\,i\left(
\begin{array}
[c]{c}%
e^{-\theta/2}\\
-e^{\theta/2}%
\end{array}
\right)  \,. \label{3.34}%
\end{equation}
In section \ref{s4} below we compare this result with the $1/N$-expansion of
the chiral $SU(N)$ Gross-Neveu model in terms of Feynman graphs.

\paragraph{The current $J_{\alpha\beta}^{\mu}(x)$:}

The $SU(N)$ current $J_{\alpha(\rho)}^{\mu}(x)$ transforms as the adjoint
representation with the weight vector $w^{J}=(2,1,\dots,1,0)$. Again, because
the Bethe ansatz yields highest weight states we obtain the matrix elements of
the highest weight component
\[
J_{\alpha(\rho)}^{\mu}=\delta_{\alpha1}\epsilon_{(\rho)N}\epsilon^{\mu\nu
}\partial_{\nu}\varphi
\]
where we have introduced the pseudo-potential $\varphi(x)$. We start from%
\begin{align*}
F_{\underline{\alpha}}^{\varphi}(\underline{\theta})  &  =K_{\underline
{\alpha}}^{\varphi}(\underline{\theta})\prod F(\theta_{ij})\\
K_{\underline{\alpha}}^{\varphi}(\underline{\theta})  &  =\int d\underline
{z}^{(1)}\,\dots\int d\underline{z}^{(N-1)}h(\underline{\theta},\underline
{\underline{z}})\,p^{\varphi}(\underline{\theta},\underline{\underline{z}%
})\,\,\Psi_{\underline{\alpha}}(\underline{\theta},\underline{\underline{z}})
\end{align*}
with $\underline{\underline{z}}=\underline{z}^{(1)},\dots,\underline
{z}^{(N-1)}$. The proposal for the p-function for $n$ particles of rank 1 (see
also \cite{BFKZ})%
\[
p^{\varphi}(\underline{\theta},\underline{\underline{z}})=N^{\varphi}\left(
\sum_{i=1}^{n}\exp\theta_{i}\right)  ^{-1}\exp\frac{1}{2}\left(  \sum
_{i=1}^{n}\theta_{i}-\sum_{i=1}^{n_{1}}z_{i}^{(1)}-\sum_{i=1}^{n_{N-1}}%
z_{i}^{(N-1)}\right)
\]
satisfies the equations (\ref{3.12}) with $w^{\varphi}=w^{J}$.

We calculate the form factor of the particle $\alpha$ and the bound state
$(\lambda)=(\lambda_{1},\dots,\lambda_{N-1})$ of $N-1$ particles with weight
vector $w=\left(  2,1,\dots,1,0\right)  $. We apply the bound state formulae
(\ref{3.16}) and (\ref{3.18}) for $n=\bar{n}=1$. Due to (\ref{2.24}) there is
no integration in each level of the nested Bethe ansatz ($l=1,\dots,N-1$) and
\begin{align}
F_{\alpha(\lambda)}^{\varphi}(\theta,\omega)  &  =K_{\alpha(\lambda)}%
^{\varphi}(\theta,\omega)\,G(\theta-\omega)\nonumber\\
K_{\alpha(\lambda)}^{\varphi}(\theta,\omega)  &  =N_{2}^{\varphi}%
\delta_{\alpha1}\epsilon_{(\lambda)N}\frac{e^{\frac{1}{2}(\theta+\omega)}%
}{e^{\theta}+e^{\omega}}\nonumber
\end{align}
The form factor for the $SU(N)$ current is therefore%
\begin{align}
F_{\alpha(\lambda)}^{J_{\beta(\rho)}^{\pm}}(\theta,\omega)  &  =\langle
\,0\,|\,J_{\beta(\rho)}^{\pm}(0)\,|\,\theta,\omega\,\rangle_{\alpha(\lambda
)}^{in}=\pm N_{2}\delta_{\alpha}^{\beta}\delta_{(\lambda)}^{(\rho)}\left(
e^{\pm\theta}+e^{\pm\omega}\right)  \frac{e^{\frac{1}{2}(\theta+\omega)}%
}{e^{\theta}+e^{\omega}}G(\theta-\omega)\nonumber\\
\,  &  =\delta_{\alpha}^{\beta}\delta_{(\lambda)}^{(\rho)}\,\bar{v}%
(\omega)\gamma^{\pm}u(\theta)\,G(\theta-\omega)/G(i\pi)\,. \label{3.36}%
\end{align}

Also here we calculate the$\ 1/N$-expansion of the exact form factor for later
comparison with the $1/N$-perturbation theory of the chiral $SU(N)$
Gross-Neveu model. Using the expansion of the minimal form factor function%
\[
G(\theta)=c^{\prime}\left(  1-\frac{1}{N}\left(  1-\frac{1}{2}\frac
{i\pi-\theta}{\tanh\frac{1}{2}\theta}\right)  \right)  +O(N^{-2})
\]
we obtain the $1/N$ expansion of the exact the $SU(N)$ current form factor as
\begin{align*}
F_{\alpha(\lambda)}^{J_{\beta(\rho)}^{\pm}}(\theta,\omega)  &  =\langle
\,0\,|\,J_{\beta(\rho)}^{\pm}(0)\,|\,\theta,\omega\,\rangle_{\alpha(\lambda
)}^{in}\,=\delta_{\alpha}^{\beta}\delta_{(\lambda)}^{(\rho)}\,\bar{v}%
(\omega)\gamma^{\pm}u(\theta)\,G(\theta-\omega)/G(i\pi)\\
&  =\delta_{\alpha}^{\beta}\delta_{(\lambda)}^{(\rho)}\,\bar{v}(\omega
)\gamma^{\pm}u(\theta)\left(  1-\frac{1}{N}\left(  1-\frac{1}{2}\frac
{i\pi-(\theta-\omega)}{\tanh\frac{1}{2}\left(  i\pi-(\theta-\omega)\right)
}\right)  \right)  +O(N^{-2})\,.
\end{align*}

\section{The chiral $SU(N)$ Gross-Neveu model}

\label{s4}Let the fermi fields $\psi_{\alpha}(x),~(\alpha=1,\dots,N)$ form an
$SU(N)$-multiplet. The field theory is defined by the Lagrangian \cite{GN}
\[
\mathcal{L}(\psi,\bar{\psi})=\bar{\psi}\,i\gamma\partial\,\psi+\tfrac{1}%
{2}g^{2}\left(  (\bar{\psi}\psi)^{2}-(\bar{\psi}\gamma^{5}\psi)^{2}\right)
\]
or equivalently%
\[
\mathcal{L}(\psi,\bar{\psi},\sigma,\pi)=\bar{\psi}(i\gamma\partial
-\sigma-i\gamma^{5}\pi)\psi-\tfrac{1}{2}g^{-2}(\sigma^{2}+\pi^{2})
\]
where $\sigma(x)$ is scalar and $\pi(x)$ a pseudoscalar field. The field
equations for these fields are%
\[
\sigma=-g^{2}\bar{\psi}\psi,~\pi=-ig^{2}\bar{\psi}\gamma^{5}\psi\,.
\]

\subsection{The 1/N perturbation theory}

Using the bootstrap program and the results of \cite{BKKW}, the S-matrix i.e.
the on-shell solution of the model has been proposed in \cite{ABW,KKS}. It is
well known \cite{GN,Wi,ABW,KKS} that the naive $1/N$-expansion of the chiral
Gross-Neveu model suffers on severe infrared problems. In \cite{ABW,KKS} two
different approaches to overcome these problems were proposed and it was shown
that the exact S-matrix was consistent with both. We will show that an
off-shell quantity as our solution for the three particle form factor of the
field $\psi(x)$ is also consistent with the $1/N$-expansion of \cite{KKS}.
Without presenting details we note that we do not obtain consistency with the
approach of \cite{ABW}. Since in the literature (see e.g. \cite{KKS,KKSe,AAR})
there are some errors and misprints we present a detailed derivation of the
approach of Swieca et al.

The generation functional of Greens's functions for the chiral Gross-Neveu
model is
\begin{equation}
Z(\xi,\bar{\xi})=\int d\psi\,d\bar{\psi}\,d\sigma\,d\pi\,\exp i\left(
\mathcal{A}(\psi,\bar{\psi},\sigma,\pi)+\bar{\xi}\psi+\bar{\psi}\xi\right)
\label{4.2}%
\end{equation}
with the action $\mathcal{A}(\psi,\bar{\psi},\sigma,\pi)=\int d^{2}%
x\,\mathcal{L}(\psi,\bar{\psi},\sigma,\pi)$. In eq.~(\ref{4.2}) and in the
following we use a short notation of the $x$-integrations e.g.~$\bar{\xi}%
\psi=\int d^{2}x\bar{\xi}(x)\psi(x)$.

When quantizing the model, severe infrared divergences appear due to the
"would-be Goldstone boson" $\pi$. Following Kurak, K\"{o}berle and Swieca
\cite{KKS} we introduce two additional bosonic fields $A(x)$ and $B(x)$
quantized with negative norm. The $A$-field compensates the infrared
divergences. In fact as we will see below that together with the infrared
divergences of $\pi$ it decouples from the rest of the model. We replace the
fermi fields by
\[
\psi(x)\rightarrow\psi^{\prime}(x)=\exp i\left(  \gamma^{5}A(x)+B(x)\right)
\,\psi(x).
\]
The $B$-field is introduced, in order not to change the statistics of the
$\psi$-fields. Finally we have the Lagrangian
\begin{align*}
\mathcal{L}  &  =\bar{\psi}^{\prime}(i\gamma\partial-\mu)\psi^{\prime}%
-\tfrac{1}{2}g^{-2}(\sigma^{2}+\pi^{2})+\tfrac{1}{2}N(\alpha^{-2}A\square
A+\beta^{-2}B\square B)\\
&  \mathrm{with}\quad\mu=\sigma+i\gamma^{5}\pi-\gamma^{5}\gamma\partial
A+\gamma\partial B\,.
\end{align*}
The couplings $\alpha$ and $\beta$ are unrenormalized, their renormalized
values are $\sqrt{\pi}$. Performing the $\psi^{\prime}$-integrations in the
generation functional we obtain
\[
Z(\xi,\bar{\xi})=\int d\sigma\,d\pi\,dA\,dB\,\exp\left(  i\mathcal{A}%
_{eff}(\sigma,\pi,A,B)-\bar{\xi}S\xi\right)
\]
with the fermi propagator $S=i(i\gamma\partial-\mu)^{-1}$ and the effective
action
\begin{multline*}
\mathcal{A}_{eff}(\sigma,\pi,A,B)\\
=-iN\operatorname*{Tr}\ln(i\gamma\partial-\mu)-\tfrac{1}{2}\int d^{2}%
x\,\left(  g^{-2}\left(  \sigma^{2}+\pi^{2}\right)  -N(\alpha^{-2}A\square
A+\beta^{-2}B\square B)\right)  \,.
\end{multline*}
The symbol $\operatorname*{Tr}$ means the trace with respect to $x$-space and
spinor space. The trace with respect to $SU(N)$-isospin has been taken and
gives the factor $N$. We define the vertex functions $\Gamma$ by
\[
\mathcal{A}_{eff}(\varphi)=\sum_{n=0}^{\infty}\frac{1}{n!}\int d^{2}x_{1}\dots
d^{2}x_{n}\,\Gamma_{\underline{\varphi}}^{(n)}(x_{1},\dots,x_{n})\,\varphi
_{1}(x_{1})\dots\varphi_{n}(x_{n})\,,
\]
where $\varphi_{i}\in\{\sigma-\sigma_{0},\pi-\pi_{0},A-A_{0},B-B_{0}\}$. The
values $\sigma_{0}$ etc.~are defined by the condition that $\mathcal{A}%
_{eff}(\varphi)$ is stationary at this point. This means that the one-point
vertex functions $\Gamma_{\sigma}^{(1)}(x)=\delta\mathcal{A}_{eff}%
/\delta\sigma=0$ etc.~vanish%
\begin{align*}
\Gamma_{\sigma}^{(1)}(x)  &  =N\,\operatorname*{tr}\,S(x,x)-g^{-2}\sigma
_{0}=0\\
\Gamma_{\pi}^{(1)}(x)  &  =N\,\operatorname*{tr}\left(  i\gamma^{5}%
S(x,x)\right)  -g^{-2}\pi_{0}=0\\
\Gamma_{A}^{(1)}(x)  &  =N\,\operatorname*{tr}\left(  -\gamma^{5}%
\gamma\partial S(x,x)\right)  +N\alpha^{-2}\square A_{0}=0\\
\Gamma_{B}^{(1)}(x)  &  =N\,\operatorname*{tr}\left(  \gamma\partial
S(x,x)\right)  +N\beta^{-2}\square B_{0}=0.
\end{align*}
The three last equations mean $\pi_{0}=A_{0}=B_{0}=0$ and the first one
implies $\sigma_{0}=m$ with
\[
N\int\frac{d^{2}p}{(2\pi)^{2}}\,\operatorname*{tr}\,\frac{i}{\gamma
p-m}-g^{-2}\sigma_{0}=0\Rightarrow\sigma_{0}=m=Me^{-\frac{\pi}{Ng^{2}}}\,,
\]
where $M$ is an UV-cutoff. There is the effect of mass generation and
dimensional transmutation: the dimensionless coupling $g$ is replaced by the
mass $m$. The $1/N$-expansion is obtained by expanding the effective action at
this stationary point. The resulting Feynman rules are given by the simple
vertices
\begin{equation}
V_{\sigma}(k)=(-i)\,,~V_{\pi}(k)=\gamma^{5}\,,~V_{A}(k)=\gamma^{5}\gamma
k\,,~V_{B}(k)=-\gamma k \label{4.4}%
\end{equation}
and the propagators in momentum space%
\begin{align}
\tilde{\Delta}_{\sigma\sigma}(k)  &  =-\frac{i\pi}{N}\,\frac{1}{\cosh
^{2}\tfrac{1}{2}\phi}\,\frac{\sinh\phi}{\phi}~,\label{4.6}\\
\tilde{\Delta}_{\pi\pi}(k)  &  =-\frac{i\pi}{N}\,\frac{1}{\sinh^{2}\tfrac
{1}{2}\phi}\left(  \frac{\sinh\phi}{\phi}-1\right)  ~,\nonumber\\
\tilde{\Delta}_{AA}(k)  &  =-\frac{i\pi}{Nk^{2}}\,~,\nonumber\\
\tilde{\Delta}_{BB}(k)  &  =-\frac{i\pi}{Nk^{2}}~,\nonumber\\
\tilde{\Delta}_{\pi A}(k)  &  =\tilde{\Delta}_{A\pi}(k)=-2m\frac{i\pi}{Nk^{2}%
}\,.\nonumber
\end{align}
where $k^{2}=-4m^{2}\sinh^{2}\tfrac{1}{2}\phi$. \begin{figure}[tbh]%
\[
\unitlength4mm\begin{picture}(2.2,4)
\put(1.75,1.8){\oval(3,2)}
\put(0,1.6){$\bullet$}
\put(-1.2,1.6){$V_i$}
\put(3,1.6){$\bullet$}
\put(3.8,1.6){$V_j$}
\put(1.3,0){$p$}
\put(1.9,.8){$\vector(1,0){0}$}
\put(1.6,2.8){$\vector(-1,0){0}$}
\put(.9,3.2){$p+k$}
\end{picture}
\]
\caption{\textit{The bubble graph.}}%
\label{f3}%
\end{figure}
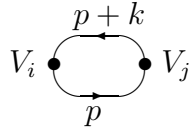To obtain the propagators one calculates the two point vertex
functions $\Gamma_{ij}^{(2)}$ from the bubble graph of Fig.\ref{f3} with the
various vertices and uses $\Delta=i{\Gamma^{(2)}}^{-1}$. In \cite{KKS} it was
argued that the unrenormalized values of $\alpha$ and $\beta$ are to be
replaced by $\alpha\rightarrow\infty$ and $\beta\rightarrow\sqrt{\pi}$. In
that limit the propagators are those of (\ref{4.6}). One observes that the
$A$- and $B$-propagators remain free and the infrared singularity in the $\pi
$-propagator disappears.

As an example we consider the four point vertex function
\[
\tilde{\Gamma}^{(4)}{}_{AB\alpha\beta}^{DC\delta\gamma}(-p_{3},-p_{4}%
,p_{1},p_{2})=\delta_{\alpha}^{\delta}\delta_{\beta}^{\gamma}\,\Gamma
_{AB}^{DC}(p_{2}-p_{3})-\delta_{\alpha}^{\gamma}\delta_{\beta}^{\delta
}\,\Gamma_{AB}^{CD}(p_{3}-p_{1})
\]
where $A,B,C,D$ are spinor indices, $\alpha,\beta,\gamma,\delta$ are isospin
indices and $\Gamma$ is given by the Feynman graph of Fig.~\ref{fa2}.
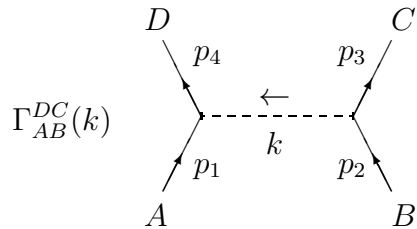
\begin{figure}[tbh]%
\[
\Gamma_{AB}^{DC}(k)%
\begin{array}
[c]{c}%
\unitlength5mm
\begin{picture}(8,6) \put(1,1){\line(1,2){1}} \put(1,1){\vector(1,2){.5}}
\put(2,3){\line(-1,2){1}} \put(2,3){\vector(-1,2){.5}}
\put(7,1){\line(-1,2){1}} \put(7,1){\vector(-1,2){.5}}
\put(6,3){\line(1,2){1}} \put(6,3){\vector(1,2){.5}}
\put(2,3){\dashbox{.2}(4,0){}} \put(3.7,2){$k$}
\put(3.5,3.3){$\leftarrow$} \put(.5,.1){$A$} \put(7,.1){$B$}
\put(7,5.3){$C$} \put(.5,5.3){$D$} \put(1.8,1.5){$p_1$}
\put(5.6,1.5){$p_2$} \put(5.6,4.5){$p_3$} \put(1.8,4.5){$p_4$}
\end{picture}
\end{array}
\]
\caption{\textit{The four point vertex }}%
\label{fa2}%
\end{figure}Taking into account the contributions from all vertices
(\ref{4.4}) and all the propagators (\ref{4.6}) we obtain
\begin{align}
\Gamma(k)  &  =\sum_{i,j}V_{i}(k)\tilde{\Delta}_{ij}(k)V_{j}(-k)\label{4.8}\\
&  =-1\otimes1\,\tilde{\Delta}_{\sigma\sigma}(k)+\gamma^{5}\otimes\gamma
^{5}\,\tilde{\Delta}_{\pi\pi}(k)-\gamma^{5}\gamma k\otimes\gamma^{5}\gamma
k\,\tilde{\Delta}_{AA}(k)\nonumber\\
&  -\gamma^{5}\otimes\gamma^{5}\gamma k\,\tilde{\Delta}_{\pi A}(k)+\gamma
^{5}\gamma k\otimes\gamma^{5}\,\tilde{\Delta}_{A\pi}(k)-\gamma k\otimes\gamma
k\,\tilde{\Delta}_{BB}(k)~.\nonumber
\end{align}
Inserting the expressions for the propagators we finally obtain
\begin{multline}
\Gamma(k)=\frac{i\pi}{N}\left\{  1\otimes1\,\frac{1}{\cosh^{2}(\phi/2)}%
\frac{\sinh\phi}{\phi}-\gamma^{5}\otimes\gamma^{5}\,\frac{1}{\sinh^{2}%
(\phi/2)}\left(  \frac{\sinh\phi}{\phi}-1\right)  \right. \label{4.10}\\
+\left.  \frac{1}{k^{2}}\left(  \gamma^{5}\gamma k\otimes\gamma^{5}\gamma
k+2m\gamma^{5}\otimes\gamma^{5}\gamma k-2m\gamma^{5}\gamma k\otimes\gamma
^{5}+\gamma k\otimes\gamma k\right)  \right\}
\end{multline}
where the tensor product structure of the spinor matrices is obvious from
Fig.~\ref{fa2}. We now apply these results to the examples of section \ref{s3}
and investigate the three particle form factor of the fundamental fermi field
and the two particle form factor of the $SU(N)$ current in $1/N$-expansion in
lowest nontrivial order.

\paragraph{The three particle form factor of the fundamental fermi field:}

For convenience we multiply the field with the Dirac operator, take
\[
{\mathcal{O}_{D\delta}(x)}=\left(  -i\left(  i\gamma\partial-m\right)
\psi(x)\right)  _{D\delta}%
\]
and define
\[
_{out}^{~~\gamma}\langle\,p_{3}\,|\,{\mathcal{O}_{D\delta}}(0)\,|\,p_{1}%
,p_{2}\,\rangle_{\alpha\beta}^{in}={F^{\mathcal{O}_{D\delta}}}_{\alpha\beta
}^{\gamma}(\theta_{12},\theta_{13},\theta_{23})\,.
\]
By means of LSZ-techniques one can express the connected part in terms of the
4-point vertex function
\[
{F_{conn.}^{\mathcal{O}_{\!D\delta}}}_{\alpha\beta}^{\gamma}(\theta
_{12},\theta_{13},\theta_{23})=\bar{u}_{C}(p_{3})\,{{\Gamma}_{AB}^{DC}%
}_{\alpha\beta}^{\delta\gamma}(-p_{3},p_{3}-p_{1}-p_{2},p_{1},p_{2}%
)\,u_{A}(p_{1})u_{B}(p_{2}).
\]
The lowest order contributions are given by the Feynman graphs of
Fig.~\ref{fa3} \begin{figure}[tbh]%
\[%
\begin{array}
[c]{c}%
\unitlength4mm
\begin{picture}(25,5) \put(2,3){\oval(4,2)}
\put(2,3){\makebox(0,0){${\cal O}_\delta(0)$}} \put(.5,1){\line(1,2){.5}}
\put(.5,1){\vector(1,2){.3}} \put(3.5,1){\line(-1,2){.5}}
\put(3.5,1){\vector(-1,2){.3}} \put(3,4){\line(1,2){.5}}
\put(3,4){\vector(1,2){.3}} \put(4,2){$\scriptstyle conn.$}
\put(.1,0){$1$} \put(3.5,0){$2$} \put(3.3,5.3){$3$} \put(6.5,2.7){$=$}
\put(8,1){\line(1,2){1}} \put(8,1){\vector(1,2){.7}}
\put(9,3){\line(-1,2){.8}} \put(9,3){\vector(-1,2){.5}}
\put(7.7,5){$\scriptstyle{\cal O}$} \put(9,3){\dashbox{.2}(2,0){}}
\put(12,1){\line(-1,2){1}} \put(12,1){\vector(-1,2){.7}}
\put(11,3){\line(1,2){1}} \put(11,3){\vector(1,2){.7}}
\put(7.7,0){$1$} \put(12,0){$2$} \put(12.2,5.3){$3$}
\put(13,2.7){$-$} \put(15,1){\line(1,1){2}} \put(15,1){\vector(1,1){.7}}
\put(18,4){\line(1,1){1}} \put(18,4){\vector(1,1){.6}}
\put(15.7,5){$\scriptstyle{\cal O}$} \put(16,2){\dashbox{.2}(3,0){}}
\put(20,1){\line(-1,1){3.8}} \put(20,1){\vector(-1,1){1.7}}
\put(14.7,0){$1$} \put(20,0){$2$} \put(19.2,5.3){$3$}
\put(21,2.7){$+\cdots$} \end{picture}
\end{array}
\]
\caption{\textit{The connected part of the three particle form factor of the
fundamental fermi field in $1/N$-expansion. }}%
\label{fa3}%
\end{figure}
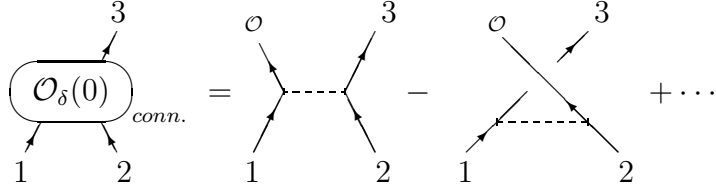%
\[
{F_{conn.}^{\mathcal{O}_{D\delta}}}_{\alpha\beta}^{\gamma}=\bar{u}_{C}%
(p_{3})\left\{  \delta_{\alpha\delta}\delta_{\beta\gamma}\,\Gamma_{AB}%
^{DC}(p_{2}-p_{3})-\delta_{\alpha\gamma}\delta_{\beta\delta}\,\Gamma_{AB}%
^{CD}(p_{3}-p_{1})\right\}  u_{A}(p_{1})u_{B}(p_{2}).
\]
where $\Gamma$ is given by Fig.~\ref{fa2} and eq.~(\ref{4.10}) and the spinor
$u(p)$ by eq.~(\ref{3.34}). It turns out that for $p_{1},~p_{2}$ and $p_{3}$
on-shell several terms vanish or cancel and we obtain up to order $1/N^{2}$%
\begin{multline}
{F_{conn.}^{\mathcal{O}_{\!D\delta}}}_{\alpha\beta}^{\gamma}=\frac{2mi\pi}%
{N}\left\{  \delta_{\alpha\delta}\delta_{\beta\gamma}\frac{\sinh\theta_{23}%
}{\theta_{23}}\left(  \frac{1}{\cosh\frac{1}{2}\theta_{23}}-\gamma^{5}\frac
{1}{\sinh\frac{1}{2}\theta_{23}}\right)  \,u_{D}(p_{1})\right. \label{4.12}\\
-\left.  \delta_{\alpha\gamma}\delta_{\beta\delta}\frac{\sinh\theta_{13}%
}{\theta_{13}}\left(  \frac{1}{\cosh\frac{1}{2}\theta_{13}}-\gamma^{5}\frac
{1}{\sinh\frac{1}{2}\theta_{13}}\right)  \,u_{D}(p_{2})\right\}
\end{multline}
which agrees with the result for the exact form factor (\ref{3.32}). In
\cite{KKS} was shown that if the momentum $p_{4}=p_{1}+p_{2}-p_{3}$ is also
on-shell then the expression (\ref{4.12}) is consistent with the exact
S-matrix (\ref{2.2}).

\paragraph{The 1/N expansion of the $SU(N)$ current form factor:}

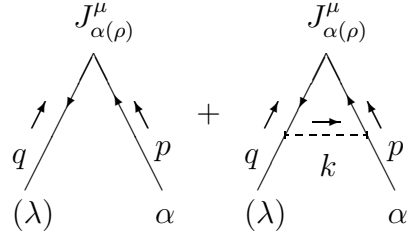
\begin{figure}[tbh]%
\[
\unitlength.9mm\begin{picture}(59,33)
\multiput(0,0)(34,0){2}{
\put(0,0){$(\lambda)$}
\put(21,0){$\alpha$}
\put(0,9){$q$}
\put(21,10){$p$}
\put(9,29){$J_{\alpha(\rho)}^{\mu}$}
\put(2,5){\line(1,2){10}}
\put(8,17){\vector(-1,-2){0}}
\put(15,19){\vector(-1,2){0}}
\put(22,5){\line(-1,2){10}}
\put(3,14){\vector(1,2){2}}
\put(20,14){\vector(-1,2){2}}}
\put(40,13){\dashbox{1}(12,0){}}
\put(27,15){+}
\put(45,7){$k$}
\put(44,15){\vector(1,0){4}}
\end{picture}
\]
\caption{Diagrams contributing to the form factor of the $SU(N)$ current in
the Gross-Neveu models up to order $N^{-2}$.}%
\label{f7}%
\end{figure}

We check the proposed exact form factor (\ref{3.36}) in $1/N$ expansion.
Fig.~\ref{f7} shows the diagrams contributing to $F_{\alpha(\lambda
)}^{J_{\beta(\rho)}^{\pm}}$ in order $N^{0}$ and $N^{-1}$ which give%
\begin{align}
F_{\alpha(\lambda)}^{J_{\beta(\rho)}^{\pm}}  &  =\delta_{\alpha}^{\beta}%
\delta_{(\lambda)}^{(\rho)}\bar{v}(q)\gamma^{\mu}u(p)F^{J}(\theta
)\label{4.14}\\
&  =\delta_{\alpha}^{\beta}\delta_{(\lambda)}^{(\rho)}\bar{v}(q)\gamma^{\pm
}u(p)+\delta_{\alpha}^{\beta}\delta_{(\lambda)}^{(\rho)}\sum_{i,j}\int
\frac{d^{2}k}{(2\pi)^{2}}\Delta_{ij}(k)\nonumber\\
&  \times\left\{  \bar{v}(q)V_{i}(k)\frac{i}{\gamma q+\gamma k-m}\gamma^{\pm
}\frac{i}{\gamma p-\gamma k-m}V_{j}(-k)u(p)-\text{substr.}\right\}
+O(N^{-2})\,. \label{4.16}%
\end{align}
The $k$ integration can be performed using the propagators (\ref{4.6}) and the
vertices (\ref{4.4}). For convenience we write the total 4-point vertex
function (\ref{4.8}) which is a part of (\ref{4.16}) as%
\[
\Gamma=\sum_{i,j}\Delta_{ij}(k)V_{i}(k)\otimes V_{j}(-k)=\Gamma_{\sigma
}+\Gamma_{\pi}+\Gamma_{V}+\Gamma_{rest}%
\]
with%
\begin{align*}
\Gamma_{\sigma}  &  =\frac{i\pi}{N}1\otimes1\,\frac{1}{\cosh^{2}(\phi/2)}%
\frac{\sinh\phi}{\phi}\\
\Gamma_{\pi}  &  =-\frac{i\pi}{N}\gamma^{5}\otimes\gamma^{5}\,\frac{1}%
{\sinh^{2}(\phi/2)}\left(  \frac{\sinh\phi}{\phi}-1\right) \\
\Gamma_{V}  &  =\frac{i\pi}{N}\gamma^{\mu}\otimes\gamma_{\mu}\\
\Gamma_{rest}  &  =\frac{i\pi}{N}\frac{1}{k^{2}}\left(  \gamma^{5}\gamma
k\otimes\gamma^{5}\left(  \gamma k-2m\right)  +\gamma^{5}\left(  \gamma
k+2m\right)  \otimes\gamma^{5}\gamma k\right)
\end{align*}
where $\gamma k\otimes\gamma k=\gamma^{5}\gamma k\otimes\gamma^{5}\gamma
k+k^{2}\gamma^{\mu}\otimes\gamma_{\mu}$ has been used. Correspondingly we
decompose the form factor function $F^{J}(\theta)$ in (\ref{4.14}) (in order
to avoid infra-red problems in the calculation) as%
\[
F^{J}(\theta)=1+\left(  F_{\sigma}(\theta)+F_{\pi}(\theta)+F_{V}%
(\theta)+F_{rest}(\theta)\right)  +O(N^{-2}).
\]
The first contribution $F_{\sigma}(\theta)$ is given by the $O(2N)$%
-Gross-Neveu form factor $F_{-}^{GN}(\theta)$ which has been calculated in
\cite{KW}%
\begin{align*}
F_{\sigma}(\theta)  &  =F_{-}^{GN}(\theta)-1\\
&  =\frac{1}{2N}\left(  \int_{0}^{\infty}\frac{\sinh^{2}\frac{1}{2}\phi}%
{\phi^{2}+\pi^{2}}\left(  \frac{\phi\coth\frac{1}{2}\phi-\hat{\theta}%
\coth\frac{1}{2}\hat{\theta}}{\cosh^{2}\frac{1}{2}\phi-\cosh^{2}\frac{1}%
{2}\hat{\theta}}-\left(  \hat{\theta}\rightarrow0\right)  \right)
d\phi\right)  +\frac{1}{N}\left(  \frac{\hat{\theta}}{\sinh\hat{\theta}%
}-1\right) \\
&  =\frac{1}{2N}\left(  1-\tfrac{1}{2}\hat{\theta}\left(  \coth\tfrac{1}%
{2}\hat{\theta}-\tanh\tfrac{1}{2}\hat{\theta}\right)  -\tfrac{1}{2}\psi\left(
\tfrac{1}{2}+\frac{\hat{\theta}}{2\pi i}\right)  -\tfrac{1}{2}\psi\left(
\tfrac{1}{2}-\frac{\hat{\theta}}{2\pi i}\right)  +\psi\left(  \tfrac{1}%
{2}\right)  \right) \\
&  +\frac{1}{N}\left(  \frac{\hat{\theta}}{\sinh\hat{\theta}}-1\right)
\end{align*}
where $\hat{\theta}=i\pi-\theta$ and $\psi\left(  z\right)  =\left(  \ln
\Gamma\left(  z\right)  \right)  ^{\prime}$. Similarly we obtain%
\begin{multline*}
F_{\pi}(\theta)=-\frac{1}{2N}\left(  \int_{0}^{\infty}\frac{\sinh^{2}\frac
{1}{2}\phi}{\phi^{2}+\pi^{2}}\left(  \frac{\phi\coth\frac{1}{2}\phi
-\hat{\theta}\coth\frac{1}{2}\hat{\theta}}{\cosh^{2}\frac{1}{2}\phi-\cosh
^{2}\frac{1}{2}\hat{\theta}}-\left(  \hat{\theta}\rightarrow0\right)  \right)
d\phi\right) \\
=-\frac{1}{2N}\left(  1-\tfrac{1}{2}\hat{\theta}\left(  \coth\tfrac{1}{2}%
\hat{\theta}-\tanh\tfrac{1}{2}\hat{\theta}\right)  -\tfrac{1}{2}\psi\left(
\tfrac{1}{2}+\frac{\hat{\theta}}{2\pi i}\right)  -\tfrac{1}{2}\psi\left(
\tfrac{1}{2}-\frac{\hat{\theta}}{2\pi i}\right)  +\psi\left(  \tfrac{1}%
{2}\right)  \right)  ,
\end{multline*}
and%
\begin{align*}
F_{V}(\theta)  &  =\frac{1}{2N}\frac{\hat{\theta}}{\sinh\hat{\theta}}\left(
\cosh\hat{\theta}-1\right) \\
F_{rest}(\theta)  &  =0
\end{align*}
and therefore%
\[
F^{J}(\theta)=1-\frac{1}{N}\left(  1-\frac{1}{2}\frac{\hat{\theta}}{\tanh
\frac{1}{2}\hat{\theta}}\right)  +O(N^{-2})
\]
which agrees with the $1/N$-expansion of the exact result for form factor of
the the current derived in section \ref{s3}.

\section{Commutation rules}

\label{s5}

In \cite{BFK} commutation rules were derived for the $Z(N)$ scaling Ising
models. The results for the $SU(N)$ Gross-Neveu model are similar, however,
the proof is much more complicated because of the unusual crossing relations
(\ref{2.38}) and (related to this) the Klein factors (\ref{2.36}).

Let $|\,\underline{\theta}\,\rangle_{\underline{\alpha}}^{in}$ with
$\underline{\alpha}=((\alpha_{11},\dots,\alpha_{1r_{1}}),\dots,(\alpha
_{\alpha1},\dots,\alpha_{\alpha r_{\alpha}}))$ be a state of $\alpha$
particles of rank $r_{1},\dots,r_{\alpha}~(1\leq r_{j}\leq N-1)$ (or bound
states of $r_{j}$ particles of rank 1). We define the charge of a state to be
the sum of all ranks of the particles in the state%
\[
Q_{\alpha}=\sum_{j=1}^{\alpha}r_{j}\,.
\]
The weight $w_{i}(\underline{\alpha})~(1\leq i\leq N)$ of the state
$\underline{\alpha}$ is equal to the number of $\alpha_{jk}=i$. Therefore the
total charge of the state $\underline{\alpha}$ is%
\[
Q_{\alpha}=\sum_{i=1}^{N}w_{i}(\underline{\alpha})
\]
(see appendix \ref{sc}). If $\underline{\alpha}$ is a state for which the form
factor $F_{\underline{\alpha}}^{\psi}(\underline{\theta})$ does not vanish we
use (\ref{3.10}) and define the charge of the operator $\psi$ by%
\begin{equation}
Q_{\psi}=Q_{\alpha}\operatorname{mod}N=\sum_{i=1}^{N}w_{i}^{\psi
}\operatorname{mod}N\nonumber
\end{equation}
with $0\leq Q_{\psi}<N$.

\subparagraph{Examples:}

For the energy momentum tensor $T^{\mu\nu}(x)$ (which is a $SU(N)$ scalar) the
fundamental field $\psi_{\alpha}(x)$ (which is a $SU(N)$ vector) and the
$SU(N)$ current $J_{\alpha\beta}^{\mu}(x)$ (which transforms as the adjoint
representation) the weights and the charges are
\[%
\begin{array}
[c]{ll}%
w^{T}=\left(  0,0,\dots,0,0\right)  , & ~Q_{T}=0\\
w^{\psi}=\left(  1,0,\dots,0,0\right)  , & ~Q_{\psi}=1\\
w^{J}=\left(  2,1,\dots,1,0\right)  , & ~Q_{J}=0\,.
\end{array}
\]

\begin{theorem}
The equal time commutation rule of two fields $\phi(x)$ and $\psi(y)$ with
charge $Q_{\phi}$ and $Q_{\psi}$, respectively, is (in general anyonic)%
\begin{equation}
\phi(x)\psi(y)=\psi(y)\phi(x)\exp\left(  2\pi i\epsilon(x^{1}-y^{1})\tfrac
{1}{2}\left(  1-1/N\right)  Q_{\phi}Q_{\psi}\right)  \,. \label{5.2}%
\end{equation}

\end{theorem}

The proof of this theorem can be found in appendix \ref{sc}.

\section{Conclusions\ }

In this paper the general $SU(N)$ form factor formula is constructed. As an
application of this result exact $SU(N)$ form factors for the field, the
energy momentum tensor and the current operators are derived in detail. In the
large $N$ limit\ these form factors are compared with the $1/N$-expansion of
the Gross-Neveu model and full agreement is found. The commutation rules of
arbitrary fields are derived and in general anyonic behavior is found. We
believe that our results may be relevant for the computation of correlation
functions in fermionic ladders \cite{KonL}. In addition the series of the
$1/N$-expansion of our exact form factors could hopefully help to understand
the same series in QCD.

\label{s6}

\subsection*{Acknowledgments}

We thank R. Schrader, B. Schroer, and A. Zapletal for useful discussions. 
In particular we thank A. Fring who participated actively in the
beginning of the SU(N)-project many years ago.
H.B. and M.K were supported by the Humboldt Foundation and H.B. also by ISTC1602.
A.F. acknowledges support from CNPq (Conselho Nacional de Desenvolvimento
Cient\'{\i}fico e Tecnol\'{o}gico). This work was also supported by the EU
network EUCLID, 'Integrable models and applications: from strings to condensed
matter', HPRN-CT-2002-00325.

\appendix

\section*{Appendix}

\section{Bound state form factors\label{sb}}

\paragraph{Proof of formulae (%
%TCIMACRO{\TeXButton{TeX field}{\protect\ref{3.18}) -- (\protect\ref{3.20}}}%
%BeginExpansion
\protect\ref{3.18}) -- (\protect\ref{3.20}%
%EndExpansion
):}

Here we present a sketch of the proof for the bound state form factors
formula. For simplicity several formulae will be written only up to constants,
the normalization can be fixed at the end by the physical properties of the
operator. The form factor formula for particles of rank 1 was proved in
\cite{BFK1}. Applying the bound state procedure to this result we derive the
formula for $\bar{n}=1$ bound state of rank $N-1$, the general case $\bar
{n}>1$ follows easily.

The bound state intertwiner \cite{BK} is defined by
\[
i\operatorname*{Res}_{\varphi_{N-1N-2}=i\eta}\dots i\operatorname*{Res}%
_{\varphi_{21}=i\eta}S_{\underline{\mu}}^{\underline{\lambda}}(\underline
{\bar{\varphi}})=\Gamma_{(\rho)}^{\underline{\lambda}}\Gamma_{\underline{\mu}%
}^{(\rho)}%
\]
where the S-matrix $S_{\underline{\mu}}^{\underline{\lambda}}(\underline
{\bar{\varphi}})$ exchanges all particles with rapidities $\underline
{\bar{\varphi}}=\varphi_{N-1},\dots,\varphi_{1}\rightarrow\underline{\varphi
}=\varphi_{1},\dots,\varphi_{N-1}$. It satisfies the bound state fusion
equation%
\begin{equation}
\Gamma_{\underline{\mu}}^{(\sigma)}S_{\underline{\lambda}\alpha}%
^{\beta\underline{\mu}}(\underline{\theta},\theta)=S_{(\rho)\alpha}%
^{\beta(\sigma)}(\omega,\theta)\Gamma_{\underline{\lambda}}^{(\rho)}\,.
\label{A.2}%
\end{equation}

\begin{lemma}
The form factor for $n$ particles $\underline{\alpha}=\alpha_{1},\dots
,\alpha_{n}$ of rank $1$ and one bound state $(\rho)=(\rho,\dots,\rho_{N-1})$
(with $\rho_{1}<\dots<\rho_{N-1}$) of rank $N-1$ may be written as%
\[
F_{\underline{\alpha}(\rho)}(\underline{\theta},\omega)=\left(  \sqrt
{2}i\right)  ^{2-N}F_{\underline{\alpha}\underline{\lambda}}(\underline
{\theta}\underline{\varphi})\Gamma_{(\rho)}^{\underline{\lambda}%
}\,,~~\text{with }\omega=\frac{1}{N-1}\left(  \varphi_{1}+\dots+\varphi
_{N-1}\right)
\]
for $\varphi_{N-1N-2}=\dots=\varphi_{21}=i\eta$, i.e. $\varphi_{j}%
=\omega+ji\eta-i\pi$.
\end{lemma}

\begin{proof}
We start with a form factor $F_{\underline{\alpha}\underline{\mu}}%
(\underline{\theta}\underline{\bar{\varphi}})$ for $n+N-1$ particles of rank
$1$ with rapidities $\underline{\theta}=\theta_{1},\dots,\theta_{n}%
,\underline{\bar{\varphi}}=\varphi_{N-1},\dots,\varphi_{1}$ and quantum
numbers $\underline{\alpha}=\alpha_{1},\dots,\alpha_{n},\underline{\mu}%
=\mu_{1},\dots,\mu_{N-1}$ (for convenience we use for $\underline{\bar
{\varphi}}$ an inverse numbering). Applying iteratively the bound state fusion
procedure (see e.g. \cite{BK,BFK1}) we obtain the bound state form factor%
\begin{align}
F_{\underline{\alpha}(\rho)}(\underline{\theta},\omega)\left(  \sqrt
{2}i\right)  ^{N-2}\Gamma_{\underline{\mu}}^{(\rho)}  &  =i\operatorname*{Res}%
_{\varphi_{N-1N-2}=i\eta}\dots i\operatorname*{Res}_{\varphi_{21}=i\eta
}F_{\underline{\alpha}\underline{\mu}}(\underline{\theta}\underline
{\bar{\varphi}})\label{A.4}\\
&  =F_{\underline{\alpha}\underline{\lambda}}(\underline{\theta}%
\underline{\varphi})i\operatorname*{Res}_{\varphi_{N-1N-2}=i\eta}\dots
i\operatorname*{Res}_{\varphi_{21}=i\eta}S_{\underline{\mu}}^{\underline
{\lambda}}(\underline{\bar{\varphi}})\nonumber\\
&  =F_{\underline{\alpha}\underline{\lambda}}(\underline{\theta}%
\underline{\varphi})\Gamma_{(\rho)}^{\underline{\lambda}}\Gamma_{\underline
{\mu}}^{(\rho)}\nonumber
\end{align}
where the form factor equation (i) $F_{\underline{\alpha}\underline{\mu}%
}(\underline{\theta}\underline{\bar{\varphi}})=F_{\underline{\alpha}%
\underline{\lambda}}(\underline{\theta}\underline{\varphi})S_{\underline{\mu}%
}^{\underline{\lambda}}(\underline{\bar{\varphi}})$ (see e.g. \cite{BFK1}) has
been used.
\end{proof}

We start from the K-function $K_{\underline{\alpha}\underline{\lambda}%
}^{\mathcal{O}}(\underline{\theta}\underline{\varphi})$ for particles of rank
1 given by the general formula (\ref{3.4}) where we replace $\underline
{\theta}\rightarrow\underline{\theta}\underline{\varphi},\,(\underline
{\varphi}=\varphi_{1},\dots,\varphi_{N-1})$ and integration variables
$\underline{z}\rightarrow\underline{z}\underline{y},\,(\underline{y}%
=y_{1},\dots,y_{N-2})$
\[
K_{\underline{\alpha}\underline{\lambda}}^{\mathcal{O}}(\underline{\theta
}\underline{\varphi})=\int d\underline{z}\,\int d\underline{y}\tilde
{h}(\underline{\theta}\underline{\varphi},\underline{z}\underline
{y})\,p(\underline{\theta}\underline{\varphi},\underline{z}\underline
{y})\,\tilde{\Psi}_{\underline{\alpha}}(\underline{\theta}\underline{\varphi
},\underline{z}\underline{y})\,.
\]
The state $\tilde{\Psi}_{\underline{\alpha}}$ is a linear combination of the
basic Bethe ansatz co-vectors (\ref{2.20}) (for $\bar{n}=0$)%
\[
\tilde{\Psi}_{\underline{\alpha}}(\underline{\theta}\underline{\varphi
},\underline{z}\underline{y})=L_{\underline{\beta}}^{(1)}(\underline
{z}\underline{y})\tilde{\Phi}_{\underline{\alpha}}^{\underline{\beta}%
}(\underline{\theta}\underline{\varphi},\underline{z}\underline{y}%
)\,,~~\text{with }1<\beta_{i}\,
\]
where the $L_{\underline{\beta}}^{(1)}(\underline{z}\underline{y})$ again
satisfy a representation like the $K_{\underline{\alpha}\underline{\lambda}%
}^{\mathcal{O}}(\underline{\theta}\underline{\varphi})$. Iterating this
nesting procedure we arrive at
\[
K_{\underline{\alpha}\underline{\lambda}}^{\mathcal{O}}(\underline{\theta
}\underline{\varphi})=\int d\underline{z}^{(1)}\int d\underline{y}%
^{(1)}\,\dots\int d\underline{z}^{(N-1)}\int\underline{y}^{(N-1)}\tilde
{h}\,p^{\mathcal{O}}\,\tilde{\Phi}_{\underline{\alpha}\underline{\lambda}}%
\]
where the functions $\tilde{h},~\,p^{\mathcal{O}}$ and $\tilde{\Phi
}_{\underline{\alpha}\underline{\lambda}}$ depend on the variables
$\underline{\theta}\underline{\varphi},\underline{\underline{z}}%
\underline{\underline{y}}$, $(\underline{\underline{z}}=\underline{z}%
^{(1)},\dots,\underline{z}^{(N-1)}$, $\underline{\underline{y}}=\underline
{y}^{(1)},\dots,\underline{y}^{(N-2)}$, $\underline{y}^{(l)}=y_{1}^{(l)}%
,\dots,y_{N-1-l}^{(l)}$, $(l=1,\dots,N-2))$. If we take the residues in
(\ref{A.4}) at $\varphi_{i+1,i}=i\eta$ the pinching phenomenon (see
\cite{BFK1} and Fig. \ref{f2}) appears at $y_{1}^{(1)}=\varphi_{1}%
,\dots,y_{N-2}^{(1)}=\varphi_{N-2}$. This propagates to the higher level
integrations such that we may replace $\underline{y}^{(l)}\rightarrow
\underline{\varphi}^{(l)}=\varphi_{1},\dots,\varphi_{N-1-l}$ which are related
to $\omega$ by $\varphi_{j}=\omega+ji\eta-i\pi$. The h-function (\ref{3.6})
for the lowest level Bethe ansatz then takes the form (up to a constant)%
\begin{align*}
\tilde{h}(\underline{\theta},\underline{\varphi},\underline{z}^{(1)}%
,\underline{\varphi}^{(1)})  &  =\tilde{\phi}(\underline{\theta}-\underline
{z})\tilde{\phi}(\underline{\theta}-\underline{\varphi}^{(1)})\tilde{\phi
}(\underline{\varphi}-\underline{z})\tau(\underline{z})\tau(\underline
{z}-\underline{\varphi}^{(1)})\\
&  =\tilde{h}(\underline{\theta},\underline{z})\tilde{\phi}(\underline{\theta
}-\underline{\varphi}^{(1)})\tilde{\phi}(\underline{\varphi}-\underline
{z})\tau(\underline{z}-\underline{\varphi}^{(1)})\,.
\end{align*}
Here and in the following we use the short notation%
\[
\tilde{\phi}(\underline{\theta}-\underline{\varphi}^{(1)})=\prod_{i=1}%
^{n}\prod_{j=1}^{N-2}\tilde{\phi}(\theta_{i}-\varphi_{j}^{(1)}),~\tau
(\underline{z})=\prod_{1\leq i<j\leq n_{1}}^{n}\tau(z_{i}-z_{j})
\]
et cetera, where the product is taken over all indices. The Bethe ansatz
states defined by (\ref{2.20}) are related by%
\[
\tilde{\Phi}_{\underline{\alpha}\underline{\lambda}}^{\underline{\beta
}\underline{\mu}}(\underline{\theta},\underline{\varphi},\underline
{z},\underline{\varphi}^{(1)})\Gamma_{(\rho)}^{\underline{\lambda}}=\tilde
{b}(\underline{\theta}-\underline{\varphi}^{(1)})\frac{b(i\pi-\omega
+\underline{z})}{a(\underline{\varphi}-\underline{z})}\Gamma_{(\sigma
)}^{\underline{\mu}1}\tilde{\Phi}_{\underline{\alpha}(\rho)}^{\underline
{\beta}(\sigma)}(\underline{\theta},\omega,\underline{z})
\]
where the bound state relation (\ref{A.2}) together with (\ref{2.8}) and
(\ref{2.6}) has been used. These equations together imply%
\[
\tilde{h}(\underline{\theta},\underline{\varphi},\underline{z},\underline
{\varphi}^{(1)})\tilde{\Phi}_{\underline{\alpha}\underline{\lambda}%
}^{\underline{\beta}\underline{\mu}}(\underline{\theta},\underline{\varphi
},\underline{z},\underline{\varphi}^{(1)})\Gamma_{(\rho)}^{\underline{\lambda
}}=\frac{\tilde{\phi}(\underline{\varphi}^{(1)}+i\eta-\underline{\theta}%
)}{\tilde{\phi}(\underline{\varphi}^{(2)}+i\eta-\underline{z})}\tilde
{h}(\underline{\theta},\underline{z})\Gamma_{(\sigma)}^{\underline{\mu}%
1}\tilde{\Phi}_{\underline{\alpha}(\rho)}^{\underline{\beta}(\sigma
)}(\underline{\theta},\omega,\underline{z})\,.
\]
The equations for bound state rapidity $\omega=\varphi_{1}-i\eta+i\pi
=\varphi_{N-1}+i\eta-i\pi$ and the relations $\tilde{b}(z-\varphi_{j}%
)\tilde{\phi}(z-\varphi_{j})=-\tilde{\phi}(\varphi_{j+1}-z)$ and $\tilde{\phi
}(\varphi_{1}-z)/\tilde{\phi}(z-\varphi_{N-1})=-b(i\pi-\omega+z)$ have been
used. Therefore we obtain the integral representation%
\begin{align*}
K_{\underline{\alpha}(\rho)}(\underline{\theta},\omega)  &  =\int
d\underline{z}\tilde{h}(\underline{\theta},\underline{z})p(\underline{\theta
},\omega,\underline{z})\tilde{\Psi}_{\underline{\alpha}(\rho)}(\underline
{\theta},\omega,\underline{z})\\
\tilde{\Psi}_{\underline{\alpha}(\rho)}(\underline{\theta},\omega
,\underline{z})  &  =L_{\underline{\beta}(\sigma)}^{(1)}(\underline{z}%
,\omega)\tilde{\Phi}_{\underline{\alpha}(\rho)}^{\underline{\beta}(\sigma
)}(\underline{\theta},\omega,\underline{z})\,,~~\text{with }1<\beta
_{i},~\sigma_{1}=1<\sigma_{2}<\dots<\sigma_{N-1}%
\end{align*}
with the new K-, L- and p-functions given in terms of the old ones
\begin{align*}
K_{\underline{\alpha}(\rho)}(\underline{\theta},\omega)  &  =\frac{1}%
{\tilde{\phi}(\underline{\varphi}^{(1)}+i\eta-\underline{\theta}%
)}K_{\underline{\alpha}\underline{\lambda}}(\underline{\theta},\underline
{\varphi})\Gamma_{(\rho)}^{\underline{\lambda}}\\
L_{\underline{\beta}(\sigma)}^{(1)}(\underline{z},\omega)  &  =\frac{1}%
{\tilde{\phi}(\underline{\varphi}^{(2)}+i\eta-\underline{z})}L_{\underline
{\beta}\underline{\mu}}^{(1)}(\underline{z}\underline{\varphi}^{(1)}%
)\Gamma_{(\sigma)}^{\underline{\mu}1}\\
p(\underline{\theta},\omega,\underline{z})  &  =p(\underline{\theta}%
\underline{\varphi},\underline{z}\underline{\varphi}^{(1)}),~\left(
\varphi_{j}=\omega+ji\eta-i\pi\right)  .
\end{align*}
Correspondingly we obtain for a higher level Bethe ansatz $l=1,\dots,N-3$%
\[
\tilde{h}(\underline{z}^{(l)},\underline{\varphi}^{(l)},\underline{z}%
^{(l+1)},\underline{\varphi}^{(l+1)})=\tilde{h}(\underline{z}^{(l)}%
,\underline{z}^{(l+1)})\tilde{\phi}(\underline{z}^{(l)}-\underline{\varphi
}^{(l+1)})\tilde{\phi}(\underline{\varphi}^{(l)}-\underline{z}^{(l+1)}%
)\tau(\underline{z}^{(l+1)}-\underline{\varphi}^{(l+1)})
\]
and%
\begin{align*}
&  \tilde{\Phi}^{(l)}{}_{\underline{\beta}\underline{\mu}}^{\underline{\gamma
}\underline{\nu}}(\underline{z}^{(l)},\underline{\varphi}^{(l)},\underline
{z}^{(l+1)},\underline{\varphi}^{(l+1)})\Gamma_{(\sigma)}^{\underline{\mu
}l\cdots1}\\
&  =\prod_{i=N-l}^{N-1}\frac{1}{b(\varphi_{i}-\underline{z}^{(l+1)})}%
\frac{\tilde{b}(\underline{z}^{(l)}-\underline{\varphi}^{(l+1)})b(i\pi
-\omega+\underline{z}^{(l+1)})}{a(\underline{\varphi}^{(l)}-\underline
{z}^{(l+1)})}\Gamma_{(\varsigma)}^{\underline{\nu}l+1\cdots1}\Phi^{(l)}%
{}_{\underline{\beta}(\sigma)}^{\underline{\gamma}(\varsigma)}(\underline
{z}^{(l)},\omega,\underline{z}^{(l+1)})
\end{align*}
such that%
\begin{multline*}
\tilde{h}(\underline{z}^{(l)},\underline{\varphi}^{(l)},\underline{z}%
^{(l+1)},\underline{\varphi}^{(l+1)})\tilde{\Phi}^{(l)}{}_{\underline{\beta
}\underline{\mu}}^{\underline{\gamma}\underline{\nu}}(\underline{z}%
^{(l)},\underline{\varphi}^{(l)},\underline{z}^{(l+1)},\underline{\varphi
}^{(l+1)})\Gamma_{(\sigma)}^{\underline{\mu}l\cdots1}\\
=\frac{\tilde{\phi}(\underline{\varphi}^{(l+1)}+i\eta-\underline{z}^{(l)}%
)}{\tilde{\phi}(\underline{\varphi}^{(l+2)}+i\eta-\underline{z}^{(l+1)}%
)}\tilde{h}(\underline{z}^{(l)},\underline{z}^{(l+1)})\Gamma_{(\varsigma
)}^{\underline{\nu}l+1\cdots1}\Phi^{(l)}{}_{\underline{\beta}(\sigma
)}^{\underline{\gamma}(\varsigma)}(\underline{z}^{(l)},\omega,\underline
{z}^{(l+1)})
\end{multline*}
and
\[
L_{\underline{\beta}(\sigma)}^{(l)}(\underline{z}^{(l)},\omega)=\int
d\underline{z}^{(l+1)}\tilde{h}(\underline{z}^{(l)},\underline{z}%
^{(l+1)})p(\underline{z}^{(l)},\underline{z}^{(l+1)})\tilde{\Psi}%
_{\underline{\beta}(\sigma)}^{(l)}(\underline{z}^{(l)},\omega,\underline
{z}^{(l+1)})
\]%
\begin{align*}
\tilde{\Psi}_{\underline{\beta}(\sigma)}^{(l)}(\underline{z}^{(l)}%
,\omega,\underline{z}^{(l+1)})  &  =L_{\underline{\gamma}(\varsigma)}%
^{(l+1)}(\underline{z}^{(l+1)},\omega)\tilde{\Phi}^{(l)}{}_{\underline{\beta
}(\sigma)}^{\underline{\gamma}(\varsigma)}(\underline{z}^{(l)},\omega
,\underline{z}^{(l+1)})\\
\,\text{with }l  &  <\gamma_{i},~\varsigma_{1}=1,\dots,\varsigma
_{l}=l<\varsigma_{l+1}<\dots<\varsigma_{N-1}\,.
\end{align*}
Note that for $l=N-3$ the relation $\tilde{\phi}(\underline{\varphi}%
^{(l+2)}+i\eta-\underline{z}^{(l+1)})=1$ holds because $\underline{\varphi
}^{(N-1)}=\emptyset$. For the highest level $l=N-2$ we have%
\begin{align*}
\tilde{h}(\underline{z}^{(N-2)},\varphi_{1},\underline{z}^{(N-1)})  &
=\tilde{h}(\underline{z}^{(N-2)},\underline{z}^{(N-1)})\tilde{\phi}%
(\varphi_{1}-\underline{z}^{(N-1)})\\
\tilde{\Phi}^{(N-2)}{}_{\underline{\beta}\mu}(\underline{z},\varphi
_{1},\underline{u})\Gamma_{(\sigma)}^{\mu N-2\cdots1}  &  =\Gamma
_{(\varsigma)}^{N-1\cdots1}\tilde{\Phi}^{(N-2)}{}_{\underline{\beta}(\sigma
)}^{(\varsigma)}(\underline{z}^{(N-2)},\omega,\underline{z}^{(N-1)})
\end{align*}
such that%
\begin{align*}
&  \tilde{h}(\underline{z}^{(N-2)},\varphi_{1},\underline{z}^{(N-1)}%
)\tilde{\Phi}^{(N-2)}{}_{\underline{\beta}\mu}(\underline{z},\varphi
_{1},\underline{u})\Gamma_{(\sigma)}^{\mu N-2\cdots1}\\
&  =\tilde{h}(\underline{z}^{(N-2)},\underline{z}^{(N-1)})\tilde{\chi}%
(\omega-\underline{z}^{(N-1)})\Gamma_{(\eta)}^{N-1\cdots1}\tilde{\Phi}%
^{(N-2)}{}_{\underline{\beta}(\sigma)}^{(\eta)}(\underline{z}^{(N-2)}%
,\omega,\underline{z}^{(N-1)})
\end{align*}
and%
\[
L_{\underline{\beta}(\sigma)}^{(N-2)}(\underline{z}^{(l)},\omega)=\int
d\underline{z}^{(N-1)}\tilde{h}(\underline{z}^{(N-2)},\underline{z}%
^{(N-1)})p(\underline{z}^{(N-2)},\underline{z}^{(N-1)})\tilde{\Psi
}_{\underline{\beta}(\sigma)}^{(N-2)}(\underline{z}^{(N-2)},\omega
,\underline{z}^{(N-1)})
\]%
\begin{align*}
\tilde{\Psi}_{\underline{\beta}(\sigma)}^{(N-2)}(\underline{z}^{(N-2)}%
,\omega,\underline{z}^{(N-1)})  &  =L^{(N-1)}(\omega-\underline{z}%
^{(N-1)})\tilde{\Phi}^{(N-2)}{}_{\underline{\beta}(\sigma)}^{\underline
{\gamma}(\eta)}(\underline{z}^{(l)},\omega,\underline{z}^{(l+1)})\\
\,\text{with }\gamma_{i}  &  =N,~\sigma_{1}=1,\dots,\sigma_{N-1}=N-1\,.
\end{align*}
Here
\[
L^{(N-1)}(\omega)=\tilde{\chi}(\omega)=\tilde{\phi}(\varphi_{1})=\Gamma\left(
\frac{1}{2}+\frac{\omega}{2\pi i}\right)  \Gamma\left(  \frac{1}{2}-\frac
{1}{N}-\frac{\omega}{2\pi i}\right)  \,.
\]

Finally we combine the minimal form factors in formula (\ref{3.8}) for
$\varphi_{N-1N-2}=\dots=\varphi_{21}=i\eta$
\begin{align*}
F_{\underline{\alpha}(\rho)}(\underline{\theta},\omega)  &  =F_{\underline
{\alpha}\underline{\lambda}}(\underline{\theta}\underline{\varphi}%
)\Gamma_{(\rho)}^{\underline{\lambda}}=F(\underline{\theta})G(\underline
{\theta}-\omega)\frac{1}{\tilde{\phi}(\underline{\varphi}+i\eta-\underline
{\theta})}K_{\underline{\alpha}\underline{\lambda}}(\underline{\theta
}\underline{\varphi})\Gamma_{(\rho)}^{\underline{\lambda}}\\
&  =F(\underline{\theta})G(\underline{\theta}-\omega)K_{\underline{\alpha
}(\rho)}(\underline{\theta},\omega)\,.
\end{align*}
The relations (\ref{2.30}) for the minimal form factor function $G$ for one
particle of rank 1 and one of rank $N-1$ and (\ref{2.32}) have been used.
Therefore the final result is%
\begin{align*}
K_{\underline{\alpha}(\rho)}^{\mathcal{O}}(\underline{\theta},\omega)  &
=\int d\underline{z}^{(1)}\dots\int d\underline{z}^{(N-1)}\tilde
{h}\,(\underline{\theta},\omega,\underline{\underline{z}})p^{\mathcal{O}%
}(\underline{\theta},\omega,\underline{\underline{z}})\tilde{\Phi}%
_{\underline{\alpha}(\rho)}(\underline{\theta},\omega,\underline{\underline
{z}})\,\\
\tilde{h}\,(\underline{\theta},\omega,\underline{\underline{z}})  &
=\prod_{l=0}^{N-2}\tilde{h}(\underline{z}^{(l)},\underline{z}^{(l+1)}%
)\prod_{i=1}^{n_{N-1}}\tilde{\chi}(\omega-z_{i}^{(N-1)})\\
p^{\mathcal{O}}(\underline{\theta},\omega,\underline{\underline{z}})  &
=p^{\mathcal{O}}(\underline{\theta}\underline{\varphi},\underline
{\underline{z}}\underline{\underline{y}})~\text{with }\underline{y}%
^{(l)}=\underline{\varphi}^{(l)}%
\end{align*}
where $p^{\mathcal{O}}(\underline{\theta}\underline{\varphi},\underline
{\underline{z}}\underline{\underline{y}})$ is the p-function for particles of
rank 1 only. The complete Bethe ansatz state is%
\[
\tilde{\Phi}_{\underline{\alpha}(\rho)}(\underline{\theta},\omega
,\underline{\underline{z}})=\tilde{\Phi}_{\underline{\varsigma}(\lambda
)}^{(N-2)(\eta)}(\underline{z}^{(N-2)},\omega,\underline{z}^{(N-1)}%
)\dots\tilde{\Phi}_{\underline{\beta}(\sigma)}^{(1)\underline{\gamma}(\kappa
)}(\underline{z}^{(1)},\omega,\underline{z}^{(2)})\tilde{\Phi}_{\underline
{\alpha}(\rho)}^{\underline{\beta}(\sigma)}(\underline{\theta},\omega
,\underline{z}^{(1)})
\]
where $(\eta)$ is the highest weight bound state $(\eta)=(1,2,\dots,N-1)$.

\paragraph{The energy momentum tensor:}

We apply the results above to the example of the energy momentum tensor and
prove (\ref{3.28}). In this case $n=\bar{n}=1,$ and the p-function is that of
(\ref{3.22})%
\[
p^{T^{\rho\sigma}}(\theta,\omega,\underline{z})=\left(  e^{\rho\theta}%
+e^{\rho\omega}\right)  e^{\sigma z}\,.
\]

\begin{lemma}
The functions $L_{\beta(\mu)}^{(l)}(z,\omega)$ (for all $l=1,\dots,N-3$) are
explicitly given as%
\begin{align}
L_{\beta(\mu)}^{(l)}(z,\omega)  &  =\epsilon_{\beta(\mu)}L^{(l)}%
(\omega-z)~~\text{with }\beta>l,~(\mu)=(1,2,\dots,l,\ast,\dots,\ast
)\nonumber\\
L^{(l)}(\omega-z)  &  =c_{l}\ \Gamma\left(  \frac{1}{2}+\frac{\omega-z}{2\pi
i}\right)  \Gamma\left(  -\frac{1}{2}+\frac{l}{N}-\frac{\omega-z}{2\pi
i}\right)  \label{A.6}%
\end{align}

\end{lemma}

\begin{proof}
Again some equations are given up to unessential constants. We use induction,
start with%
\[
L^{(N-1)}(z,\omega)=\epsilon_{N(1\dots N-1)}\tilde{\chi}(\omega-z)=(-1)^{N-1}%
\Gamma\left(  \frac{1}{2}+\frac{\omega-z}{2\pi i}\right)  \Gamma\left(
\frac{1}{2}-\frac{1}{N}-\frac{\omega-z}{2\pi i}\right)  \,,
\]
and then calculate iteratively for $l=N-1,\dots,2$ the integrals%
\begin{align*}
L_{\beta(\mu)}^{(l-1)}(z,\omega)  &  =\int_{\mathcal{C}_{z\omega}}%
du\,\tilde{\phi}(z-u)L^{(l)}(\omega-u)\epsilon_{\gamma(\nu)}{\tilde{\Phi}%
}_{\beta(\mu)}^{(l-1)\gamma(\nu)}(z,\omega,u)\\
\epsilon_{\gamma(\nu)}{\tilde{\Phi}}_{\beta(\mu)}^{(l-1)\gamma(\nu)}%
(z,\omega,u)  &  =\epsilon_{\gamma(\nu)}\tilde{S}_{\beta\delta}^{\gamma
l}(z-u)\tilde{S}_{(\mu)l}^{\delta(\nu)}(\omega-u)\\
&  =\epsilon_{\beta(\mu)}\left(  \left(  N-l\right)  \delta_{\beta}^{l}%
\tilde{b}(z-u)\tilde{d}(\omega-u)+\delta_{\beta}^{>l}\tilde{c}(z-u)\right)
\end{align*}
where $\delta_{\beta}^{>l}=1$ for $\beta>l$ and $0$ else. Both integrals
\begin{align*}
I_{1}  &  =\int_{\mathcal{C}_{z}\omega}du\tilde{\phi}(z-u)L^{(l)}%
(\omega-u)(N-l)\tilde{b}(z-u)\tilde{d}(\omega-u)\\
I_{2}  &  =\int_{\mathcal{C}_{z}\omega}du\tilde{\phi}(z-u)L^{(l)}%
(\omega-u)\tilde{c}(z-u)
\end{align*}
can be calculated by means of the formula%
\begin{multline*}
\int_{-\infty}^{\infty}dz\Gamma\left(  a+\frac{z}{2\pi i}\right)
\Gamma\left(  b+\frac{z}{2\pi i}\right)  \Gamma\left(  c-\frac{z}{2\pi
i}\right)  \Gamma\left(  d-\frac{z}{2\pi i}\right) \\
=\left(  2\pi\right)  ^{2}\frac{\Gamma\left(  a+c\right)  \Gamma\left(
a+d\right)  \Gamma\left(  b+c\right)  \Gamma\left(  b+d\right)  }%
{\Gamma\left(  a+b+c+d\right)  }%
\end{multline*}
and yield the result (\ref{A.6}).
\end{proof}

Finally we have to calculate
\begin{align}
K_{\alpha(\lambda)}(\theta,\omega)  &  =\left(  e^{\rho\theta}+e^{\rho\omega
}\right)  \int_{\mathcal{C}_{\theta\omega}}dz\,\tilde{\phi}(\theta
-z)L^{(1)}(\omega-z)e^{\sigma z}\epsilon_{\delta(\mu)}{\tilde{\Phi}}%
_{\alpha(\lambda)}^{\delta(\mu)}(\theta,\omega,z)\,\nonumber\\
\epsilon_{\delta(\mu)}{\tilde{\Phi}}_{\alpha(\lambda)}^{\delta(\mu)}%
(\theta,\omega,z)\,  &  =\epsilon_{\alpha(\lambda)}\left(  \left(  N-1\right)
\delta_{\alpha}^{1}\tilde{b}(\theta-z)\tilde{d}(\omega-z)+\delta_{\alpha}%
^{>1}\tilde{c}(\theta-z)\right) \nonumber
\end{align}
which yields the result (\ref{3.28}) using the formula%
\[
\int_{\mathcal{C}}\left(  \Gamma(a+\frac{z}{2\pi i})\Gamma(b+\frac{z}{2\pi
i})\Gamma(c-\frac{z}{2\pi i})\Gamma(d-\frac{z}{2\pi i})\right)  e^{\sigma
z}dz=\frac{\sigma\left(  2\pi i\right)  ^{3}}{ab-cd}\exp\left(  \sigma
i\pi\left(  c+d\right)  \right)
\]
for $a+b+c+d=0$.

\section{Commutation rules\label{sc}}

%\label{ac}
In this appendix we use the short notation for form factors i.e. matrix
elements of the field $\psi(x)$ at $x=0$%
\begin{equation}
\psi_{\underline{\alpha}}^{\underline{\beta}}(\underline{\theta}_{\beta
}^{\prime},\underline{\theta}_{\alpha})=F^{\psi}\,_{\underline{\alpha}%
}^{\underline{\beta}}(\underline{\theta}_{\beta}^{\prime},\underline{\theta
}_{\alpha})=\,^{\underline{\beta},out}\langle\,\underline{\theta}_{\beta
}^{\prime}|\,\psi(0)|\,\underline{\theta}_{\alpha}\,\rangle_{\underline
{\alpha}}^{in}\,. \label{B.2}%
\end{equation}
To proof the general commutation rules of fields (\ref{5.2}) we have to
consider $SU(N)$ sum rules and general crossing relations.

\subsection{$SU(N)$ sum rules}

\subparagraph{Particles and anti-particles:}

Let $(\alpha)=(\alpha_{1},\dots,\alpha_{r}),~(1\leq\alpha_{1}<\dots<\alpha
_{r}\leq N)$ a particle of rank (and charge) $r$. The corresponding
anti-particle is $(\bar{\alpha})=(\bar{\alpha}_{1},\dots,\bar{\alpha}_{N-r})$,
$(1\leq\bar{\alpha}_{1}<\dots<\bar{\alpha}_{N-r}\leq N)$ (of rank $N-r$) such
that the union of the set of indices satisfies $\{\alpha_{1},\dots,\alpha
_{r}\}\cup\{\bar{\alpha}_{1},\dots,\bar{\alpha}_{N-r}\}=\{1,\dots,N\}$.
Therefore
\begin{equation}
\sum_{k=1}^{r}\alpha_{k}+\sum_{k=1}^{N-r}\bar{\alpha}_{k}=\sum_{k=1}%
^{N}k=\frac{1}{2}N\left(  N+1\right)  . \label{B.4}%
\end{equation}

\subparagraph{Charges:}

Let $\underline{\alpha}=((\alpha_{11},\dots,\alpha_{1r_{1}}),\dots
,(\alpha_{\alpha1},\dots,\alpha_{\alpha r_{\alpha}}))$ be a state of $\alpha$
particles of rank $r_{1},\dots,r_{\alpha}~(1\leq r_{j}\leq N-1)$ (or bound
states of $r_{j}$ particles of rank 1). We define the charge of a state as the
sum of all ranks of the particles in the state $\underline{\alpha}$%
\[
Q_{\alpha}=\sum_{j=1}^{\alpha}r_{j}\,.
\]
The charge of anti-particles (bound states) we define as
\begin{align*}
Q_{\bar{\alpha}}  &  =\sum_{j=1}^{\alpha}r_{j}\left(  N-1\right)  =\left(
N-1\right)  Q_{\alpha}\\
Q_{\alpha}+Q_{\bar{\alpha}}  &  =NQ_{\alpha}\,.
\end{align*}

\subparagraph{Weights:}

Let $\underline{\alpha}=((\alpha_{11},\dots,\alpha_{1r_{1}}),\dots
,(\alpha_{\alpha1},\dots,\alpha_{\alpha r_{\alpha}}))$ be a state of $\alpha$
particles. The weight $w_{i}(\underline{\alpha})$ of the state $\underline
{\alpha}$ is equal to the number of $\alpha_{jk}=i,~(1\leq i\leq N)$
\[
w_{i}(\underline{\alpha})=\sum_{j=1}^{\alpha}\sum_{k=1}^{r_{j}}\delta
_{i\alpha_{jk}}\,,~(1\leq i\leq N).
\]
Therefore the total charge of the state $\underline{\alpha}$ is%
\[
Q_{\alpha}=\sum_{j=1}^{\alpha}r_{j}=\sum_{j=1}^{\alpha}\sum_{k=1}^{r_{j}%
}1=\sum_{j=1}^{\alpha}\sum_{k=1}^{r_{j}}\sum_{i=1}^{N}\delta_{i\alpha_{jk}%
}=\sum_{i=1}^{N}w_{i}(\underline{\alpha})
\]
Similarly, we consider $\underline{\gamma}=((\gamma_{11},\dots,\gamma_{1s_{1}%
}),\dots,(\gamma_{\gamma1},\dots,\gamma_{\gamma s_{\gamma}}))$.

\subparagraph{Sum rules:}

Because of $SU(N)$ invariance%
\[
\psi_{\underline{\alpha}}^{\underline{\gamma}}(\underline{\theta}_{\gamma
}^{\prime},\underline{\theta}_{\alpha})=\,^{\underline{\gamma},out}%
\langle\,\underline{\theta}_{\gamma}^{\prime}|\,\psi(0)|\,\underline{\theta
}_{\alpha}\,\rangle_{\underline{\alpha}}^{in}\neq0
\]
(or $\psi_{\underline{\alpha}\underline{\bar{\gamma}}}\neq0$) implies for the
weights
\[
w(\underline{\alpha})=w(\underline{\gamma})+w^{\psi}+L\left(  1,\dots
,1\right)  ,~L\in\mathbb{Z}%
\]
where $w^{\psi}$ is the weight vector of the operator $\psi$ and $\left(
1,\dots,1\right)  $ are weights of a state in the vacuum sector. Therefore%
\[
Q_{\alpha}=\sum_{i=1}^{N}w_{i}(\underline{\alpha})=Q_{\gamma}+\sum_{i=1}%
^{N}w_{i}^{\psi}+NL\,.
\]
The charge of the operator $\psi$ is defined by%
\begin{align}
Q_{\psi}  &  =\left(  Q_{\alpha}-Q_{\gamma}\right)  \operatorname{mod}N,~0\leq
Q_{\psi}<N\label{B.6}\\
&  =\sum_{i=1}^{N}w_{i}^{\psi}\operatorname{mod}N\,.\nonumber
\end{align}
For a particle $(\alpha_{j1},\dots,\alpha_{jr_{j}})$ of rank $r_{j}$ we use
the short notation $(\alpha_{j})=(\alpha_{j1},\dots,\alpha_{jr_{j}})$ and
$\alpha_{j}=\sum_{k=1}^{r_{j}}\alpha_{jk}$ and correspondingly, $\gamma
_{j}=\sum_{k=1}^{s_{j}}\gamma_{jk}$. Then $SU(N)$ invariance implies%
\begin{align}
\sum_{j=1}^{\alpha}\alpha_{j}-\sum_{j=1}^{\gamma}\gamma_{j}  &  =\tfrac{1}%
{2}(N+1)\left(  Q_{\alpha}-Q_{\gamma}-Q_{\psi}\right)  +R_{\psi}\nonumber\\
\text{with~~~~~~}R_{\psi}  &  =\sum_{i=1}^{N}i\,w_{i}^{\psi}-\tfrac{1}%
{2}(N+1)\left(  \sum_{i=1}^{N}w_{i}^{\psi}-Q_{\psi}\right)  \label{B.8}%
\end{align}
which can be straightforwardly proved using the above definitions.

Examples:%
\[%
\begin{array}
[c]{llll}%
T^{\mu\nu}: & w^{T}=\left(  0,0,\dots,0\right)  & Q_{T}=0 & R_{T}=0\\
\psi_{\alpha}: & w^{\psi}=\left(  1,0,\dots,0\right)  & Q_{\psi}=1 & R_{\psi
}=1\\
j^{\mu\nu}: & w^{J}=\left(  2,1,\dots,1,0\right)  & Q_{j}=0 & R_{j}=1-N
\end{array}
\]

\subsection{Crossing}

\subsubsection{A partial S-matrix}

\begin{definition}
\textbf{ } Let $\underline{\theta}_{\beta}=(\theta_{\pi(1)},\dots,\theta
_{\pi(\alpha)})$ be a permutation of $\underline{\theta}_{\alpha}=(\theta
_{1},\dots,\theta_{\alpha})$. Then $S_{\underline{\alpha}}^{\underline{\beta}%
}\left(  \underline{\theta}_{\beta};\underline{\theta}_{\alpha}\right)  $ is
the matrix representation of the permutation group $\mathcal{S}_{\alpha}$
generated by the simple transpositions $\sigma_{ij}:i\leftrightarrow j$ for
any pair of nearest neighbor indices $1\leq i,j=i+1\leq\alpha$
as\footnote{Note that this definition is quite analogous to that of
representations of the braid group by means of spectral parameter independent
R-matrices.}
\[
\sigma_{ij}\rightarrow S(\theta_{ij})\,
\]
Because of the Yang-Baxter relation and unitarity of the S-matrix the
representation is well defined. We will also use the notation
\[
S_{\underline{\alpha}}^{\underline{\mu}\underline{\lambda}}\left(
\underline{\theta}_{\mu}\underline{\theta}_{\lambda};\underline{\theta
}_{\alpha}\right)
\]
if $\pi$ is that permutation which reorders the array $\underline{\theta
}_{\alpha}$ such that it coincides with the combined arrays of $\underline
{\theta}_{\mu}$ and $\underline{\theta}_{\lambda}$.
\end{definition}

As an example consider the case $\underline{\theta}_{\alpha}=(\theta
_{1},\theta_{2},\theta_{3},\theta_{4}),~\underline{\theta}_{\mu}=(\theta
_{2},\theta_{3})$ and $\underline{\theta}_{\lambda}=(\theta_{1},\theta_{4})$
\begin{align*}
S_{\underline{\alpha}}^{\underline{\mu}\underline{\lambda}}(\theta_{2}%
\theta_{3}\theta_{1}\theta_{4};\theta_{1}\theta_{2}\theta_{3}\theta_{4})  &
=S_{\alpha_{1}^{\prime}\alpha_{3}}^{\mu_{2}\lambda_{1}}(\theta_{13}%
)S_{\alpha_{1}\alpha_{2}}^{\mu_{1}\alpha_{1}^{\prime}}(\theta_{12}%
)\delta_{\alpha_{4}}^{\lambda_{2}}\\%
\begin{array}
[c]{l}%
\unitlength6mm\begin{picture}(4,4) \put(0,1.5){\framebox(4,1){$S$}}
\put(3.5,.7){\line(0,1){.8}} \put(3.5,2.5){\line(0,1){.8}}
\put(.5,.7){\line(0,1){.8}} \put(.5,2.5){\line(0,1){.8}}
\put(1.5,.7){\line(0,1){.8}} \put(1.5,2.5){\line(0,1){.8}}
\put(2.5,.7){\line(0,1){.8}} \put(2.5,2.5){\line(0,1){.8}}
\put(.3,.2){$\alpha_1$} \put(.6,.8){$\theta_1$} \put(1.3,.2){$\alpha_2$}
\put(1.6,.8){$\theta_2$} \put(2.3,.2){$\alpha_3$} \put(2.6,.8){$\theta_3$}
\put(3.3,.2){$\alpha_4$} \put(3.6,.8){$\theta_4$} \put(.6,2.7){$\theta_2$}
\put(1.6,2.7){$\theta_3$} \put(2.6,2.7){$\theta_1$} \put(3.6,2.7){$\theta_4$}
\put(.3,3.5){$\mu_1$} \put(1.3,3.5){$\mu_2$} \put(2.3,3.5){$\lambda_1$}
\put(3.3,3.5){$\lambda_1$}
\end{picture}
\end{array}
&  =%
\begin{array}
[c]{l}%
\unitlength6mm\begin{picture}(4,3) \put(0,1){\line(2,1){2}}
\put(1,1){\line(-1,1){1}} \put(2,1){\line(-1,1){1}} \put(3,1){\line(0,1){1}}
\put(-.2,.3){$\alpha_1$} \put(0.8,.3){$\alpha_2$} \put(1.8,.3){$\alpha_3$}
\put(2.8,.3){$\alpha_4$} \put(-.2,2.4){$\mu_1$} \put(0.8,2.4){$\mu_2$}
\put(1.8,2.4){$\lambda_1$} \put(2.8,2.4){$\lambda_2$}
\end{picture}
\end{array}
\end{align*}
If the permutation inverts the rapidities completely $S_{\underline{\alpha}%
}^{\underline{\beta}}\left(  \underline{\theta}_{\beta};\underline{\theta
}_{\alpha}\right)  =S_{\underline{\alpha}}^{\underline{\beta}}(\underline
{\theta}_{\alpha})$ is the full S-matrix.

\subsubsection{Crossing for SU(N)}

As was argued by Swieca et al.~\cite{KKS} the particles of the chiral $SU(N)$
Gross-Neveu model posses anyonic statistics and due to the unusual crossing
property of the S-matrix, Klein factors are needed. The crossing relations of
the form factors for normal fields and particles were derived in \cite{BK} by
means of LSZ-assumptions and maximal analyticity. They have to be modified for
the chiral $SU(N)$ Gross-Neveu model.

We propose crossing relations%
\begin{align}
\psi_{\underline{\alpha}}^{\underline{\gamma}}(\underline{\theta}_{\gamma
}^{\prime};\underline{\theta}_{\alpha})  &  =\sigma_{(\gamma)}^{\psi}\sum_{%
%TCIMACRO{\QTATOP{\theta_{\eta}\cup\theta_{\nu}=\theta_{\gamma}}{\theta
%_{\lambda}\cup\theta_{\mu}=\theta_{\alpha}}}%
%BeginExpansion
\genfrac{}{}{0pt}{1}{\theta_{\eta}\cup\theta_{\nu}=\theta_{\gamma}%
}{\theta_{\lambda}\cup\theta_{\mu}=\theta_{\alpha}}%
%EndExpansion
}\zeta_{(\gamma,\alpha,\eta)}^{\psi}S_{\underline{\nu}\underline{\eta}%
}^{\underline{\gamma}}(\underline{\theta}_{\gamma};\underline{\theta}_{\nu
}\underline{\theta}_{\eta})\,\mathbf{1}_{\underline{\mu}}^{\underline{\nu}%
}\,\mathbf{C}^{\underline{\eta}\underline{\bar{\eta}}}\psi_{\underline
{\bar{\eta}}\underline{\lambda}}(\underline{\theta}_{\bar{\eta}}^{\prime}%
+i\pi_{-},\underline{\theta}_{\lambda})S_{\underline{\alpha}}^{\underline{\mu
}\underline{\lambda}}\left(  \underline{\theta}_{\mu}\underline{\theta
}_{\lambda};\underline{\theta}_{\alpha}\right) \label{B.10}\\
&  =\sum_{%
%TCIMACRO{\QTATOP{\theta_{\eta}\cup\theta_{\nu}=\theta_{\gamma}}{\theta
%_{\lambda}\cup\theta_{\mu}=\theta_{\alpha}}}%
%BeginExpansion
\genfrac{}{}{0pt}{1}{\theta_{\eta}\cup\theta_{\nu}=\theta_{\gamma}%
}{\theta_{\lambda}\cup\theta_{\mu}=\theta_{\alpha}}%
%EndExpansion
}\xi_{(\gamma,\alpha,\eta)}^{\psi}S_{\underline{\eta}\underline{\nu}%
}^{\underline{\gamma}}(\underline{\theta}_{\gamma};\underline{\theta}_{\eta
}\underline{\theta}_{\nu})\,\mathbf{C}^{\underline{\bar{\eta}}\underline{\eta
}}\psi_{\underline{\lambda}\underline{\bar{\eta}}}\,\mathbf{1}_{\underline
{\mu}}^{\underline{\nu}}S_{\underline{\alpha}}^{\underline{\lambda}%
\underline{\mu}}\left(  \underline{\theta}_{\lambda}\underline{\theta}_{\mu
};\underline{\theta}_{\alpha}\right)  \label{B.12}%
\end{align}
which, compared to the formulae in \cite{BK}, are modified by the factors
$\zeta_{(\gamma,\alpha,\eta)}^{\psi}$ and $\xi_{(\gamma,\alpha,\eta)}^{\psi}$%
\begin{align*}
\zeta_{(\gamma,\alpha,\eta)}^{\psi}  &  =\rho_{(\gamma,\alpha)}^{\psi}%
e^{i\pi\left(  N-1\right)  \frac{1}{2}Q_{\eta}\left(  Q_{\eta}+N\right)
}e^{\frac{2\pi i}{N}\left(  R_{\psi}Q_{\alpha}-Q_{\psi}\sum\bar{\gamma}%
_{j}\right)  }\\
\xi_{(\gamma,\alpha,\eta)}^{\psi}  &  =e^{i\pi\left(  N-1\right)  \left(
\frac{1}{2}Q_{\eta}\left(  Q_{\eta}+N\right)  +Q_{\nu}Q_{\psi}\right)
}e^{\frac{2\pi i}{N}\left(  R_{\psi}Q_{\alpha}-Q_{\psi}\sum\bar{\gamma}%
_{j}\right)  }\\
\rho_{(\gamma,\alpha)}^{\psi}  &  =\left(  -1\right)  ^{\left(
N-1+(1-1/N)\left(  Q_{\alpha}+Q_{\bar{\gamma}}-Q_{\psi}\right)  \right)
Q_{\bar{\gamma}}}\\
\sigma_{(\gamma)}^{\psi}  &  =e^{i\pi(1-1/N)Q_{\psi}Q_{\bar{\gamma}}}%
\end{align*}
with $\bar{\gamma}_{j}=\frac{1}{2}N(N+1)-\gamma_{j}$, due to (\ref{B.4}). The
sign factor $\rho_{(\gamma,\alpha)}^{\psi}$ and the statistics factor
$\sigma_{(\gamma)}^{\psi}$ were introduced in \cite{BFK1}. The charge
$Q_{\psi}$ of the operator $\psi$ and the number $R_{\psi}$ are defined in
(\ref{B.6}) and (\ref{B.8}).

\subsection{Commutation rules}

In \cite{BFK} commutation rules were derived for the $Z(N)$ scaling Ising
models. The results for the $SU(N)$ Gross-Neveu model are very similar,
however the proof is more complicated because of the unusual crossing
relations and the presence of the Klein factors.

\begin{theorem}
The equal time commutation rule of two fields $\phi(x)$ and $\psi(y)$ with
charge $Q_{\phi}$ and $Q_{\psi}$, respectively, is (in general anyonic)%
\[
\phi(x)\psi(y)=\psi(y)\phi(x)\exp\left(  2\pi i\epsilon(x^{1}-y^{1})\tfrac
{1}{2}\left(  1-1/N\right)  Q_{\phi}Q_{\psi}\right)  \,.
\]

\end{theorem}

\begin{proof}
We consider an arbitrary matrix element of products of fields
\[
\left(  \phi(x)\psi(y)\right)  _{\underline{\alpha}}^{\underline{\beta}%
}(\underline{\theta}_{\beta}^{\prime},\underline{\theta}_{\alpha
})=\,^{\underline{\beta},out}\langle\,\underline{\theta}_{\beta}^{\prime
}|\,\phi(x)\psi(y)|\,\underline{\theta}_{\alpha}\,\rangle_{\underline{\alpha}%
}^{in}\,.
\]
Inserting a complete set of intermediate states $|\,\underline{\tilde{\theta}%
}_{\gamma}\,\rangle_{\underline{\gamma}}^{in}$ we obtain%
\begin{equation}
\left(  \phi(x)\psi(y)\right)  _{\underline{\alpha}}^{\underline{\beta}%
}(\underline{\theta}_{\beta}^{\prime},\underline{\theta}_{\alpha
})=e^{iP_{\beta}^{\prime}x-iP_{\alpha}y}\frac{1}{\gamma!}\int_{\underline
{\tilde{\theta}}_{\gamma}}\phi_{\underline{\gamma}}^{\underline{\beta}%
}(\underline{\theta}_{\beta}^{\prime},\underline{\tilde{\theta}}_{\gamma}%
)\psi_{\underline{\alpha}}^{\underline{\gamma}}(\underline{\tilde{\theta}%
}_{\gamma},\underline{\theta}_{\alpha})e^{-i\tilde{P}_{\gamma}(x-y)}
\label{B.14}%
\end{equation}
where $P_{\alpha}=$ the total momentum of the state $|\,\underline{\theta
}_{\alpha}\,\rangle_{\underline{\alpha}}^{in}$ etc. and $\int_{\underline
{\tilde{\theta}}_{\gamma}}=\prod_{k=1}^{\gamma}\int\frac{d\tilde{\theta}_{k}%
}{4\pi}$. Einstein summation convention over all sets $\underline{\gamma}$ is
assumed. We also define $\gamma!=\prod_{r=1}^{N}n_{r}!$ where $n_{r}$ is the
number of particles of rank $r$ in $\underline{\gamma}$. We apply the general
crossing formulae (\ref{B.10},\ref{B.12}). Strictly speaking, we apply the
second version (\ref{B.12}) of the crossing formula to the matrix element of
$\phi$
\[
\phi_{\underline{\gamma}}^{\underline{\beta}}(\underline{\theta}_{\beta
}^{\prime},\underline{\tilde{\theta}}_{\gamma})=\sum_{\substack{\underline
{\theta}_{\rho}^{\prime}\cup\underline{\theta}_{\tau}^{\prime}=\underline
{\theta}_{\beta}^{\prime}\\\underline{\tilde{\theta}}_{\varsigma}%
\cup\underline{\tilde{\theta}}_{\sigma}=\underline{\tilde{\theta}}_{\gamma}%
}}\xi_{(\beta,\gamma,\rho)}^{\phi}S_{\underline{\rho}\underline{\tau}%
}^{\underline{\beta}}\,\phi_{\underline{\varsigma}\underline{\bar{\rho}}%
}(\underline{\tilde{\theta}}_{\varsigma},\underline{\theta^{\prime}}%
_{\bar{\rho}}-i\pi_{-})\mathbf{C}^{\underline{\rho}\underline{\bar{\rho}}%
}\,\mathbf{1}_{\underline{\sigma}}^{\underline{\tau}}\,S_{\underline{\gamma}%
}^{\underline{\varsigma}\underline{\sigma}}%
\]
where $\underline{\bar{\rho}}=\left(  \bar{\rho}_{\rho},\dots,\bar{\rho}%
_{1}\right)  $ with $\bar{\rho}=$ antiparticle of $\rho$ and $\underline
{\theta^{\prime}}_{\bar{\rho}}-i\pi_{-}$ means that all rapidities taken the
values $\theta^{\prime}-i\left(  \pi-\epsilon\right)  $. The matrix
$\mathbf{1}_{\underline{\sigma}}^{\underline{\tau}}(\underline{\theta^{\prime
}}_{\tau},\underline{\tilde{\theta}}_{\varkappa})$ is defined by (\ref{B.2})
with $\mathcal{O}=\mathbf{1}$ the unit operator. Summation is over all
decompositions of the sets of rapidities $\underline{\theta}_{\beta}^{\prime}$
and $\underline{\tilde{\theta}}_{\gamma}$. To the matrix element of $\psi$ we
apply the first version of the crossing formula (\ref{B.10})
\[
\psi_{\underline{\alpha}}^{\underline{\gamma}}(\underline{\tilde{\theta}%
}_{\gamma},\underline{\theta}_{\alpha})=\sigma_{(\gamma,\alpha)}^{\psi}\sum_{%
%TCIMACRO{\QTATOP{\underline{\theta}_{\nu}\cup\underline{\theta}_{\eta
%}=\underline{\theta}_{\gamma}}{\underline{\theta}_{\mu}\cup\underline{\theta
%}_{\lambda}=\underline{\theta}_{\alpha}}}%
%BeginExpansion
\genfrac{}{}{0pt}{1}{\underline{\theta}_{\nu}\cup\underline{\theta}_{\eta
}=\underline{\theta}_{\gamma}}{\underline{\theta}_{\mu}\cup\underline{\theta
}_{\lambda}=\underline{\theta}_{\alpha}}%
%EndExpansion
}\zeta_{(\gamma,\alpha,\eta)}^{\psi}S_{\underline{\nu}\underline{\eta}%
}^{\underline{\gamma}}\,\mathbf{1}_{\underline{\mu}}^{\underline{\nu}%
}\,\mathbf{C}^{\underline{\eta}\underline{\bar{\eta}}}\psi_{\underline
{\bar{\eta}}\underline{\lambda}}(\underline{\theta}_{\bar{\eta}}^{\prime}%
+i\pi_{-},\underline{\theta}_{\lambda})S_{\underline{\alpha}}^{\underline{\mu
}\underline{\lambda}}\,.
\]
We insert (\ref{B.10}) and (\ref{B.12}) in (\ref{B.14}) and use the product
formula $S_{\underline{\gamma}}^{\underline{\varsigma}\underline{\sigma}%
}(\underline{\tilde{\theta}}_{\varsigma}\underline{\tilde{\theta}}_{\sigma
};\underline{\tilde{\theta}}_{\gamma})S_{\underline{\nu}\underline{\eta}%
}^{\underline{\gamma}}(\underline{\tilde{\theta}}_{\gamma};\underline
{\tilde{\theta}}_{\nu}\underline{\tilde{\theta}}_{\eta})=S_{\underline{\nu
}\underline{\eta}}^{\underline{\varsigma}\underline{\sigma}}(\underline
{\tilde{\theta}}_{\varsigma}\underline{\tilde{\theta}}_{\sigma};\underline
{\tilde{\theta}}_{\nu}\underline{\tilde{\theta}}_{\eta})$. Let us first assume
that the sets rapidities in the initial state $\underline{\theta}_{\alpha}$
and the ones of the final state $\underline{\theta}_{\beta}^{\prime}$ have no
common elements. This also implies $\underline{\tilde{\theta}}_{\nu}%
\cap\underline{\tilde{\theta}}_{\sigma}=\emptyset$. Then we may use (ii) to
get $S_{\underline{\nu}\underline{\eta}}^{\underline{\varsigma}\underline
{\sigma}}(\underline{\tilde{\theta}}_{\varsigma}\underline{\tilde{\theta}%
}_{\sigma};\underline{\tilde{\theta}}_{\nu}\underline{\tilde{\theta}}_{\eta
})=\delta_{\underline{\nu}\underline{\pi}}^{\underline{\varsigma}%
\underline{\sigma}}$ and then we can perform the $\underline{\tilde{\theta}%
}_{\nu}$- and $\underline{\tilde{\theta}}_{\sigma}$-integrations. The
remaining $\tilde{\theta}$-integration variables are $\underline{\tilde
{\theta}}_{\omega}=\underline{\tilde{\theta}}_{\varsigma}\cap\underline
{\tilde{\theta}}_{\eta}$. Then we may write for the sets of particles
$\underline{\varsigma}=\underline{\mu}\underline{\omega},\,\underline{\eta
}=\underline{\omega}\underline{\tau}$ and $\underline{\gamma}=\underline{\mu
}\underline{\omega}\underline{\tau}$ and similar for rapidities and momenta.
Equation (\ref{B.14}) simplifies to%
\begin{multline}
\left(  \phi(x)\psi(y)\right)  _{\underline{\alpha}}^{\underline{\beta}%
}(\underline{\theta}_{\beta}^{\prime},\underline{\theta}_{\alpha}%
)=\sum_{\substack{\underline{\theta}_{\rho}^{\prime}\cup\underline{\theta
}_{\tau}^{\prime}=\underline{\theta}_{\beta}^{\prime}\\\underline{\theta}%
_{\mu}\cup\underline{\theta}_{\lambda}=\underline{\theta}_{\alpha}}}\frac
{\mu!\tau!}{\mu\omega\tau!}S_{\underline{\rho}\underline{\tau}}^{\underline
{\beta}}(\underline{\theta}_{\rho}^{\prime},\underline{\theta}_{\tau}^{\prime
})\int_{\underline{\tilde{\theta}}_{\omega}}X_{\underline{\mu}\underline
{\lambda}}^{\underline{\rho}\underline{\tau}}\label{B.16}\\
\times S_{\underline{\alpha}}^{\underline{\mu}\underline{\lambda}}%
(\underline{\theta}_{\alpha})e^{i\left(  P_{\rho}^{\prime}-P_{\mu}\right)
x-i\left(  P_{\lambda}-P_{\tau}^{\prime}\right)  y}%
\end{multline}
where%
\begin{multline}
X_{\underline{\mu}\underline{\lambda}}^{\underline{\rho}\underline{\tau}%
}=\sigma_{(\gamma,\alpha)}^{\psi}\zeta_{(\gamma,\alpha,\eta)}^{\psi}%
\xi_{(\beta,\gamma,\rho)}^{\phi}\phi_{\underline{\mu}\underline{\omega
}\underline{\bar{\rho}}}(\underline{\theta}_{\mu},\underline{\tilde{\theta}%
}_{\omega},\underline{\theta^{\prime}}_{\bar{\rho}}-i\pi_{-})\label{B.18}\\
\times\mathbf{C}^{\underline{\bar{\rho}}\underline{\rho}}\mathbf{C}%
^{\underline{\tau}\underline{\bar{\tau}}}\mathbf{C}^{\underline{\omega
}\underline{\bar{\omega}}}\psi_{\underline{\bar{\tau}}\underline{\bar{\omega}%
}\underline{\lambda}}(\underline{\theta}_{\bar{\tau}}^{\prime}+i\pi
_{-},\underline{\tilde{\theta}}_{\bar{\omega}}+i\pi_{-},\underline{\theta
}_{\lambda})e^{-i\tilde{P}_{\omega}(x-y)}\,.
\end{multline}

Similarly, if we apply for the operator product $\psi(y)\phi(x)$ and%
\[
\left(  \psi(y)\phi(x)\right)  _{\underline{\alpha}}^{\underline{\beta}%
}(\underline{\theta}_{\beta}^{\prime},\underline{\theta}_{\alpha
})=e^{iP_{\beta}^{\prime}x-iP_{\alpha}y}\frac{1}{\delta!}\int_{\underline
{\tilde{\theta}}_{\delta}}\psi_{\underline{\delta}}^{\underline{\beta}%
}(\underline{\theta}_{\beta}^{\prime},\underline{\tilde{\theta}}_{\delta}%
)\phi_{\underline{\alpha}}^{\underline{\delta}}(\underline{\tilde{\theta}%
}_{\delta},\underline{\theta}_{\alpha})e^{-i\tilde{P}_{\delta}(y-x)}\,,
\]
use the second crossing formula to the matrix element of $\phi$
\[
\phi_{\underline{\alpha}}^{\underline{\delta}}(\underline{\tilde{\theta}%
}_{\delta},\underline{\theta}_{\alpha})=\sum_{\substack{\underline
{\tilde{\theta}}_{\varphi}\cup\underline{\tilde{\theta}}_{\kappa}%
=\underline{\tilde{\theta}}_{\delta}\\\underline{\theta}_{\mu}\cup
\underline{\theta}_{\lambda}=\underline{\theta}_{\alpha}}}\xi_{(\delta
,\alpha,\varphi)}^{\phi}S_{\underline{\varphi}\underline{\kappa}}%
^{\underline{\delta}}\,\phi_{\underline{\mu}\underline{\bar{\varphi}}%
}(\underline{\theta}_{\mu},\underline{\tilde{\theta}}_{\bar{\varphi}}-i\pi
_{-})\mathbf{C}^{\underline{\varphi}\underline{\bar{\varphi}}}\,\mathbf{1}%
_{\underline{\lambda}}^{\underline{\kappa}}\,S_{\underline{\alpha}%
}^{\underline{\mu}\underline{\lambda}}%
\]
and the first one to the matrix element of $\psi$
\[
\psi_{\underline{\delta}}^{\underline{\beta}}(\underline{\theta}_{\beta
}^{\prime},\underline{\tilde{\theta}}_{\delta})=\dot{\sigma}_{(\beta,\delta
)}^{\psi}\sum_{%
%TCIMACRO{\QTATOP{\underline{\theta}_{\rho}^{\prime}\cup\underline{\theta
%}_{\tau}^{\prime}=\underline{\theta}_{\beta}^{\prime}}{\underline
%{\tilde{\theta}}_{\xi}\cup\underline{\tilde{\theta}}_{\chi}=\underline
%{\tilde{\theta}}_{\delta}}}%
%BeginExpansion
\genfrac{}{}{0pt}{1}{\underline{\theta}_{\rho}^{\prime}\cup\underline{\theta
}_{\tau}^{\prime}=\underline{\theta}_{\beta}^{\prime}}{\underline
{\tilde{\theta}}_{\xi}\cup\underline{\tilde{\theta}}_{\chi}=\underline
{\tilde{\theta}}_{\delta}}%
%EndExpansion
}\zeta_{(\beta,\delta,\tau)}^{\psi}S_{\underline{\rho}\underline{\tau}%
}^{\underline{\beta}}\,\mathbf{1}_{\underline{\xi}}^{\underline{\rho}%
}\,\mathbf{C}^{\underline{\tau}\underline{\bar{\tau}}}\psi_{\underline
{\bar{\tau}}\underline{\chi}}(\underline{\theta}_{\bar{\tau}}^{\prime}%
+i\pi_{-},\underline{\tilde{\theta}}_{\chi})S_{\underline{\delta}}%
^{\underline{\xi}\underline{\chi}}\,.
\]
Similarly we obtain (with $\underline{\chi}=\underline{\bar{\omega}}%
\underline{\lambda},\,\underline{\varphi}=\underline{\bar{\omega}}%
\underline{\rho},~\underline{\delta}=\underline{\mu}\underline{\bar{\omega}%
}\underline{\tau}$) equation (\ref{B.16}) where $X_{\underline{\mu}%
\underline{\lambda}}^{\underline{\rho}\underline{\tau}}$ replaced by%
\begin{multline}
Y_{\underline{\mu}\underline{\lambda}}^{\underline{\rho}\underline{\tau}%
}=\sigma_{(\beta,\delta)}^{\psi}\zeta_{(\beta,\delta,\tau)}^{\psi}\xi
_{(\delta,\alpha,\varphi)}^{\phi}\phi_{\underline{\mu}\underline{\omega
}\underline{\bar{\rho}}}(\underline{\theta}_{\mu},\underline{\tilde{\theta}%
}_{\omega}-i\pi_{-},\underline{\theta^{\prime}}_{\bar{\rho}}-i\pi
_{-})\label{B.20}\\
\times\mathbf{C}^{\underline{\bar{\rho}}\underline{\rho}}\mathbf{C}%
^{\underline{\tau}\underline{\bar{\tau}}}\mathbf{C}^{\underline{\omega
}\underline{\bar{\omega}}}\psi_{\underline{\bar{\tau}}\underline{\bar{\omega}%
}\underline{\lambda}}(\underline{\theta}_{\bar{\tau}}^{\prime}+i\pi
_{-},\underline{\tilde{\theta}}_{\bar{\omega}},\underline{\theta}_{\lambda
})e^{i\tilde{P}_{\omega}(x-y)}%
\end{multline}
which means that only $\sigma_{(\gamma,\alpha)}^{\psi}\zeta_{(\gamma
,\alpha,\eta)}^{\psi}\xi_{(\beta,\gamma,\rho)}^{\phi}$ is replaced by
$\sigma_{(\beta,\delta)}^{\psi}\zeta_{(\beta,\delta,\tau)}^{\psi}\xi
_{(\delta,\alpha,\varphi)}^{\phi}$ and the integration variables
$\underline{\tilde{\theta}}_{\omega}$ by $\underline{\tilde{\theta}}%
_{\bar{\omega}}-i\pi_{-}$, i.e. $\tilde{P}_{\omega}$ by $-\tilde{P}_{\omega}$.

If there were no bound states, there would be no singularities in the physical
strip and we could shift in the matrix element of $\psi(y)\phi(x)$
(\ref{B.16}) with (\ref{B.20}) for equal times and $x^{1}<y^{1}$ the
integration variables by $\tilde{\theta}_{i}\rightarrow\tilde{\theta}_{i}%
+i\pi_{-}$. Note that the factor $e^{i\tilde{P}_{\omega}(x-y)}$ decreases for
$0<\operatorname{Re}\tilde{\theta}_{i}<\pi$ if $x^{1}<y^{1}$. Because
$\tilde{P}_{\omega}\rightarrow-\tilde{P}_{\omega}$ (if $\underline
{\tilde{\theta}}_{\omega}\rightarrow\underline{\tilde{\theta}}_{\bar{\omega}%
}-i\pi_{-}$) we get the matrix element of $\phi(x)\psi(y)$ (\ref{B.16}) with
(\ref{B.18}) up to the factor%
\[
\frac{\sigma_{(\gamma,\alpha)}^{\psi}\zeta_{(\gamma,\alpha,\eta)}^{\psi}%
\xi_{(\beta,\gamma,\rho)}^{\phi}}{\sigma_{(\beta,\delta)}^{\psi}\zeta
_{(\beta,\delta,\tau)}^{\psi}\xi_{(\delta,\alpha,\varphi)}^{\phi}}=e^{-2\pi
i\frac{1}{2}\left(  1-\frac{1}{N}\right)  Q_{\phi}Q_{\psi}}\,.
\]
This equality follows after a long and cumbersome but straightforward
calculation. In \cite{BFK} was shown that we obtain the same result if there
are bound states.
\end{proof}

%\bibliographystyle{phreport}
%\bibliography{../ref}

\begin{thebibliography}{99}


\bibitem {Hooft1}G.~'t~Hooft,
\newblock {A planar diagram theory for strong interactions}, \newblock Nuclear
Physics B \textbf{72}, 461--473 (1974).

\bibitem {Hooft2}G.~'t~Hooft, \newblock {A two dimensional model for mesons},
\newblock Nuclear Physics B \textbf{75}, 461--470 (1974).

\bibitem {Mano}A.~V. Manohar, \newblock {Large N QCD}, \newblock (1998).

\bibitem {GN}D.~J. Gross and A.~Neveu, \newblock Dynamical symmetry breaking
in asymptotically free field theories, \newblock Phys. Rev. \textbf{D10},
3235 (1974).

\bibitem {Wi}E.~Witten, \newblock Chiral symmetry, the 1/N expansion, and the
SU(N) Thirring model, \newblock Nucl. Phys. \textbf{B145}, 110 (1978).

\bibitem {KuS}V.~Kurak and J.~A. Swieca, \newblock Anti-particles as bound
states of particles in the factorized S-matrix framework, \newblock Phys.
Lett. \textbf{B82}, 289--291 (1979).

\bibitem {ABW}E.~Abdalla, B.~Berg, and P.~Weisz, \newblock More about the
S-matrix of the chiral SU(N) Thirring model, \newblock Nucl. Phys.
\textbf{B157}, 387--391 (1979).

\bibitem {KKS}R.~Koberle, V.~Kurak, and J.~A. Swieca, \newblock Scattering
theory and 1/N expansion in the chiral Gross- Neveu model, \newblock Phys.
Rev. \textbf{D20}, 897--902 (1979).

\bibitem {AAR}E.~Abdalla, M.~C.~B. Abdalla, and K.~D. Rothe,
\newblock {Nonperturbative methods in two-dimensional quantum field theory},
\newblock Singapore, Singapore: World Scientific (2001) p.~832.

\bibitem {K}M.~Karowski, \newblock Exact S matrices and form-factors in
(1+1)-dimensional field theoretic models with soliton behavior,
\newblock Phys. Rept. \textbf{49}, 229--237 (1979).

\bibitem {K2}M.~Karowski, \newblock The bootstrap program for 1+1 dimensional
field theoretic models with soliton behavior, \newblock in: W. R{\"u}hl (Ed.),
Field theoretic methods in particle physics, Plenum, New York, (1980) ,
307--324, \newblock Presented at Kaiserslautern NATO Inst. 1979.

\bibitem {KW}M.~Karowski and P.~Weisz, \newblock Exact form factors in
(1+1)-dimensional field theoretic models with soliton behavior,
\newblock Nucl. Phys. \textbf{B139}, 455--476 (1978).

\bibitem {BFK1}H.~M. Babujian, A.~Foerster, and M.~Karowski,
\newblock {The nested SU(N) off-shell Bethe ansatz and exact form factors},
\newblock J. Phys. \textbf{A41}, 275202 (2008).

\bibitem {Sm}F.~Smirnov, \newblock Form Factors in Completely Integrable
Models of Quantum Field Theory, \newblock Adv. Series in Math. Phys.
\textbf{14}, World Scientific (1992).

\bibitem {NT}A.~Nakayashiki and Y.~Takeyama, \newblock On form factors of
SU(2) invariant Thirring model, \newblock math-ph/0105040 (2001).

\bibitem {Ta}Y.~Takeyama, \newblock Form factors of SU(N) invariant Thirring
model, \newblock Publ. Res. Inst. Math. Sci. Kyoto \textbf{39}, 59--116
(2003).

\bibitem {MTV}E.~Mukhin, V.~Tarasov, and A.~Varchenko, \newblock Bethe
Eigenvectors of Higher Transfer Matrices, \newblock math.QA/0605015 (2006).

\bibitem {Pa}S.~Pakuliak, \newblock Weight Functions and Nested Bethe Ansatz,
\newblock Ann. Hernri Poincare \textbf{7}, 1541--1554 (2007).

\bibitem {BKKW}B.~Berg, M.~Karowski, V.~Kurak, and P.~Weisz,
\newblock Factorized U(n) symmetric S matrices in two-dimensions,
\newblock Nucl. Phys. \textbf{B134}, 125--132 (1978).

\bibitem {BW}B.~Berg and P.~Weisz, \newblock Exact S-matrix of the chiral
invariant SU(N) Thirring model, \newblock Nucl. Phys. \textbf{B146},
205--214 (1978).

\bibitem {K1}M.~Karowski, \newblock On the bound state problem in (1+1)
dimensional field theories, \newblock Nucl. Phys. \textbf{B153},
244--252 (1979).

\bibitem {B1}H.~M. Babujian, \newblock Correlation function in WZNW model as a
Bethe wave function for the Gaudin magnetics, \newblock In 'Gosen 1990,
Proceedings, Theory of elementary particles' 12-23. (see high energy physics
index 29 (1991) No. 12257) , 12--23 (1990).

\bibitem {BKZ2}H.~Babujian, M.~Karowski, and A.~Zapletal, \newblock Matrix
Difference Equations and a Nested Bethe Ansatz, \newblock J. Phys.
\textbf{A30}, 6425--6450 (1997).

\bibitem {FST}L.~D. Faddeev, E.~K. Sklyanin, and L.~A. Takhtajan,
\newblock The quantum inverse problem method. 1, \newblock Theor. Math. Phys.
\textbf{40}, 688--706 (1980).

\bibitem {TF}L.~A. Takhtajan and L.~D. Faddeev, \newblock The Quantum method
of the inverse problem and the Heisenberg XYZ model, \newblock Russ. Math.
Surveys \textbf{34}, 11--68 (1979).

\bibitem {KT1}M.~Karowski and H.~J. Thun,
\newblock {Complete S matrix of the O(2N) Gross-Neveu model}, \newblock Nucl.
Phys. \textbf{B190}, 61--92 (1981).

\bibitem {BFKZ}H.~M. Babujian, A.~Fring, M.~Karowski, and A.~Zapletal,
\newblock Exact form factors in integrable quantum field theories: The
sine-Gordon model, \newblock Nucl. Phys. \textbf{B538}, 535--586 (1999).

\bibitem {BK}H.~Babujian and M.~Karowski, \newblock Exact form factors in
integrable quantum field theories: The sine-Gordon model. II, \newblock Nucl.
Phys. \textbf{B620}, 407--455 (2002).

\bibitem {BK2}H.~Babujian and M.~Karowski, \newblock Sine-Gordon form factors
and quantum field equations, \newblock J. Phys. \textbf{A35}, 9081--9104
(2002).

\bibitem {BK04}H.~Babujian and M.~Karowski, \newblock Exact form factors for
the scaling Z(N)-Ising and the affine A(N-1) Toda quantum field theories,
\newblock Phys. Lett. \textbf{B575}, 144--150 (2003).

\bibitem {BFK}H.~Babujian, A.~Foerster, and M.~Karowski, \newblock Exact form
factors in integrable quantum field theories: The scaling Z(N)-Ising model,
\newblock Nucl. Phys. \textbf{B736}, 169--198 (2006).

\bibitem {BKZ1}H.~Babujian, M.~Karowski, and A.~Zapletal, \newblock U(N)
Matrix Difference Equations and a Nested Bethe Ansatz,
\newblock hep-th/9611006 (1996).

\bibitem {KKSe}R.~Koberle, V.~Kurak, and J.~A. Swieca, \newblock Erratum:
Scattering theory and 1/N expansion in the chiral Gross- Neveu model,
\newblock Phys. Rev. \textbf{D20}, 2638 (1979).

\bibitem {KonL}R.~Konik and A.~W.~W. Ludwig,
\newblock {Exact zero temperature correlation functions for two leg Hubbard
ladders and carbon nanotubes}, \newblock cond-mat/9810332 (1998).
\end{thebibliography}

\end{document}